\PassOptionsToPackage{nopieces}{cls/core}
\PassOptionsToPackage{nobblmodule}{cls/core}

\documentclass{fmt/addendum/class}

\usepackage[margin=1.25in]{geometry}

\usepackage{xr}
\usepackage{src/temp-preamble}

\makeatletter

\def\chapterautorefname{Chapter}

\def\chapterautorefname{\S\@gobble}

\def\sectionautorefname{\S\@gobble}

\def\subsectionautorefname{\S\@gobble}

\def\subsubsectionautorefname{\S\@gobble}

\def\appendixautorefname{\S\@gobble}

\makeatother
\usepackage{authblk}

\newcommand{\appref}[1]{\autoref{app:#1}}
\newcommand{\srdsref}[1]{\autoref{#1} (\cite[\autoref*{srds-#1}]{srds-version})}

\newabbreviation[%
    type=notation,%
    category=notation-category,%
    description=A subset of vertices in an undirected graph such that all are pairwise-not-neighbors
]
{IS}{Independent-Set}{Independent Set of Vertices in an Undirected Graph}

\newabbreviation[%
    category=notation-category,%
    description=A valid coloring for vertices in an undirected graph
]
{coloring}{coloring}{Vertex Coloring}
\newglossaryentry{symb:chromatic}{%
    type=notation,%
    category=notation-category,%
    name=$\chromaticnumberformat{}{G}$,%
    description={The chromatic number of a graph $G$ (Unweighted)},
text={}
}

\newglossaryentry{symb:weighted-chromatic}{%
type=notation,%
category=notation-category,%
name=$\chromaticnumberformat{\ell}{G}$,%
description={The chromatic number of a graph $G$ with vertex weighting fuction $\ell$},
text={}
}

\newglossaryentry{symb:haspath}{%
type=notation,%
category=notation-category,%
name={$v \longrightsquiglearrow_G u$},%
description={The graph $G$ contains a path from vertex $u$ to the vertex $v$.},
text={}
}

\newglossaryentry{symb:disjointunion}{%
type=notation,%
category=notation-category,%
name={$\bigsqcup$},%
description={Disjoint Union},
text={}
}

\newglossaryentry{symb:nodedepth}{%
type=notation,%
category=notation-category,%
name={$\nodedepthformat{G\,}{v}$},%
description={The is the number of vertices in a maximal simple path in graph $G$ that ends at vertex $v$.},
text={}
}

\newglossaryentry{symb:depth}{%
type=notation,%
category=notation-category,%
name={$\depthformat{\big[\ell\big]}{G}$},%
description={The depth of a graph $G$ is the number of vertices in a maximal simple path in the graph. When given a weight function $\ell$ for vertices, we use the sum of the weight of the vertices in the path instead of the number of vertices.},
text={}
}

\newglossaryentry{symb:N}{%
type=notation,%
category=notation-category,%
name=\ensuremath{\mathbb{N}},%
description={The set of natural numbers}
}

\newabbreviation[%
type=notation,%
category=notation-category,%
description={}
]
{NPH}{NP-hard}{Nondeterministic-Polynomial Hard}

\newabbreviation[%
type=notation,%
category=notation-category,%
description={}
]
{NPC}{NP-Complete}{Nondeterministic-Polynomial Complete}

\newabbreviation[%
category=notation-category,%
]
{Greedy Schedule}{Level Schedule Function}{The \textsc{LevelSchedule} function in \autoref*{alg:greedy-schedule}}

\newglossaryentry{symb:N+}{%
type=notation,%
category=notation-category,%
name=\ensuremath{\mathbb{N}^+},%
description={
The set of natural numbers \emph{excluding} 0
(\(\mathbb{N}^+ := \mathbb{N}\setminus \left\{0\right\}\))
}
}

\newglossaryentry{symb:N0}{%
type=notation,%
category=notation-category,%
name=\ensuremath{\mathbb{N}^0},%
description={
The set of natural numbers \emph{including} 0
(\(\mathbb{N}^0 = \mathbb{N} \cup \left\{0\right\}\))
}
}

\newglossaryentry{symb:IS}{%
type=notation,%
category=notation-category,%
name=\ensuremath{\mathfrak{ISet}},%
description={The Independent Set Decision Problem \vfill
\(\mathfrak{ISet}=\left\{\left(G, k\right) \in \text{Graphs} \times \mathbb{N}\right\}\)}
}

\newglossaryentry{symb:Color}{%
type=notation,%
category=notation-category,%
name=\ensuremath{Color},%
description={Vertex coloring decision problem (see~\autoref{def:color})}
}

\newglossaryentry{symb:MinColor}{%
type=notation,%
category=notation-category,%
name=\ensuremath{MinColor},%
description={Vertex coloring with a minimal number of colors search problem (see~\autoref{def:mincolor})}
}

\newglossaryentry{symb:MinColoring}{%
type=notation,%
category=notation-category,%
name=\ensuremath{MinColoring},%
description={Minimal vertex coloring search problem (see~\autoref{def:mincoloring})}
}

\newglossaryentry{symb:OptimalLatency}{%
name=\ensuremath{OptimalLatency},%
description=,
}

\newglossaryentry{symb:Latency}{%
name=\ensuremath{Latency},%
description=,
}

\newglossaryentry{symb:OptimalSchedule}{%
name=\ensuremath{OptimalSchedule},%
description=,
}

\newglossaryentry{symb:NPH}{%
type=notation,%
category=notation-category,%
name=\ensuremath{NPH},%
description=Set of all NP-Hard Problems
}

\newglossaryentry{symb:NPC}{%
type=notation,%
category=notation-category,%
name=\ensuremath{NPC},%
description=Set of all NP-Complete Decision Problems
}

\newacronym[
category=notation-category,
description=A directed graph with no paths that visit the same vertex twice
]{DAG}{DAG}{directed acyclic graph}

\newacronym[
category=notation-category,
description=
]{wlog}{w.l.o.g}{Without loss of generality}

\newacronym[
category=notation-category,
description=
]{tx/s}{tx/s}{transactions per second}

\newacronym[
category=notation-category,
description=
]{ASMR}{ASMR}{Active State Machine Replication}

\newacronym[
category=notation-category,
description=
]{SMR}{SMR}{State Machine Replication}

\newabbreviation[%
category=notation-category,%
description=
]
{BRSchlrA}{Block-Runtime}{Block-Runtime Algorithm}

\newacronym[%
category=notation-category,%
description=
]
{scheduler}{Block-Runtime}{Block-Runtime Algorithm}

\newabbreviation[%
category=notation-category,%
description=A set of transactions in which no two transactions conflict
]
{conflict-free}{conflict-free}{sets with no conflicting transactions}

\newabbreviation[%
category=notation-category,%
description=
]
{sequentially deterministic}{sequentially deterministic}{}

\newacronym[%
category=notation-category,%
description=
]
{schedule}{schedule}{Dependency Schedule}

\newabbreviation[%
category=notation-category,%
description=
]
{scheduling graph}{scheduling graph}{Dependency Scheduling Graph}

\newglossary[trm]{term}{trmin}{trmout}{Terms}

\newglossaryentry{serializable}{%
type=term,%
category=term-category,%
name=serializable,%
description=
}

\newglossaryentry{serializability}{%
type=term,%
category=term-category,%
name=serializability,%
description=
}

\newabbreviation[%
category=notation-category,%
description=
]
{homogeneous transactions}{homogeneous transactions}{transactions with homogeneous execution times}
\newabbreviation[%
category=notation-category,%
description=
]
{heterogeneous transactions}{heterogeneous transactions}{transactions with heterogeneous execution times}
\usepackage{pkg/orcidlink}

\makeatletter
\newcommand{\clonelabel}[2]{\@bsphack 
  \expandafter\ifx\csname r@#2\endcsname\relax
  \else\protected@write\@auxout{}{\string\newlabel{#1}
    {\csname r@#2\endcsname}}
  \fi
  \expandafter\ifx\csname r@#2@cref\endcsname\relax 
  \else\protected@write\@auxout{}{\string\newlabel{#1@cref}
    {\csname r@#2@cref\endcsname}}
  \fi
  \@esphack} 
\makeatother
\begin{document}
    \title{Batch-Schedule-Execute: On Optimizing Concurrent Deterministic Scheduling for Blockchains (Extended Version)}
    \author[1]{Yaron Hay}
    \author[2]{Roy Friedman}
    \affil[1,2]{Computer Science Faculty, Technion, Israel}
    \affil[1]{\href{mailto:yaron.hay@cs.technion.ac.il}{yaron.hay@cs.technion.ac.il} \orcidlinkc{0009-0006-1263-7318}}
    \affil[2]{\href{mailto:roy@cs.technion.ac.il}{roy@cs.technion.ac.il} \orcidlinkc{0000-0001-6460-9665}}
    \date{}
    \maketitle
\begin{abstract}
    Executing smart contracts is a compute and storage-intensive task, which currently dominates modern blockchain's performance.
    Given that computers are becoming increasingly multicore, concurrency is an attractive approach to improve programs' execution runtime.
    A unique challenge of blockchains is that all replicas (miners or validators) must execute all smart contracts in the same logical order to maintain the semantics of \gls{SMR}.

    In this work, we study the maximal level of parallelism attainable when focusing on the conflict graph between transactions packaged in the same block.
    This exposes a performance vulnerability that block creators may exploit against existing blockchain concurrency solutions, which rely on a total ordering phase for maintaining consistency amongst all replicas.
    To facilitate the formal aspects of our study, we develop a novel generic framework for \gls{ASMR} that is strictly serializable.
    We introduce the concept of graph scheduling and the definition of the minimal latency scheduling problem, which we prove to be \glsxtrshort{NPH}.
    We show that the restricted version of this problem for homogeneous transactions is equivalent to the classic Graph Vertex Coloring Problem, yet show that the heterogeneous case is more complex.
    We discuss the practical implications of these results.
\end{abstract}
    \clearpage
    \section{Introduction}
    \label{chap:intro}

    Smart contracts~\cite{SzaboSC, EthereumWhitepaper} are used in many blockchains to enable rich semantics required by the FinTech industry, as well as to realize the Web 3.0 vision~\cite{EthereumYellowpaper}.
    Yet, in many modern blockchains, local smart contracts' execution and validation serves as a major performance bottleneck.
    This is true for both HyperLedger Fabric~\cite{FastFabric} as well as other blockchains whose consensus-based ordering mechanism can support anywhere between 10K-1M \gls{tx/s}, but local smart contracts execution, in contrast, is limited to as low as hundreds of \gls{tx/s}~\cite{FastFabric, ParBlockchain, Diablo, SRBB, BlockSTM}.

    Most efforts to reduce smart contract execution and validation times have revolved around improving the system's design aspects and improving the software engineering of the virtual machine implementation~\cite{FastFabric, SRBB}.
    A few works have also addressed applying parallelism to local smart contracts execution~\cite{HerlihySC, ParBlockchain, BlockSTM}, which is critical due to the fact that modern CPUs are highly parallel architectures.
    As we know from DB theory, ensuring an ordering only among conflicting transactions is enough to guarantee \gls{serializability}~\cite{DB-textbook}.
    Alas, the main difference between concurrent smart contracts and database concurrency control~\cite{DBCC} is that in the former, all replicas must ensure the \emph{same} logical total order on all transactions, which requires the serialization order to be deterministic.
    Previous concurrency efforts in blockchains, and state machine replication in general~\cite{Schneider1990}, relied on the total ordering induced from the consensus protocol to derive a unique ordering between conflicting transactions.
    This is commonly known as the \emph{Order-Execute} Paradigm (XO).
    However, as we show in \autoref{fig:gap}, there are scenarios where following the consensus-induced ordering leads to poor utilization, up to completely sequential executions.
    This is despite the fact that, in this example, serialization can be obtained for the same workload with a very high level of~concurrency.

    \begin{figure}
        \centering
        \begin{subfigure}{0.3095\linewidth}
            \centering
            \includegraphics[width=\linewidth]{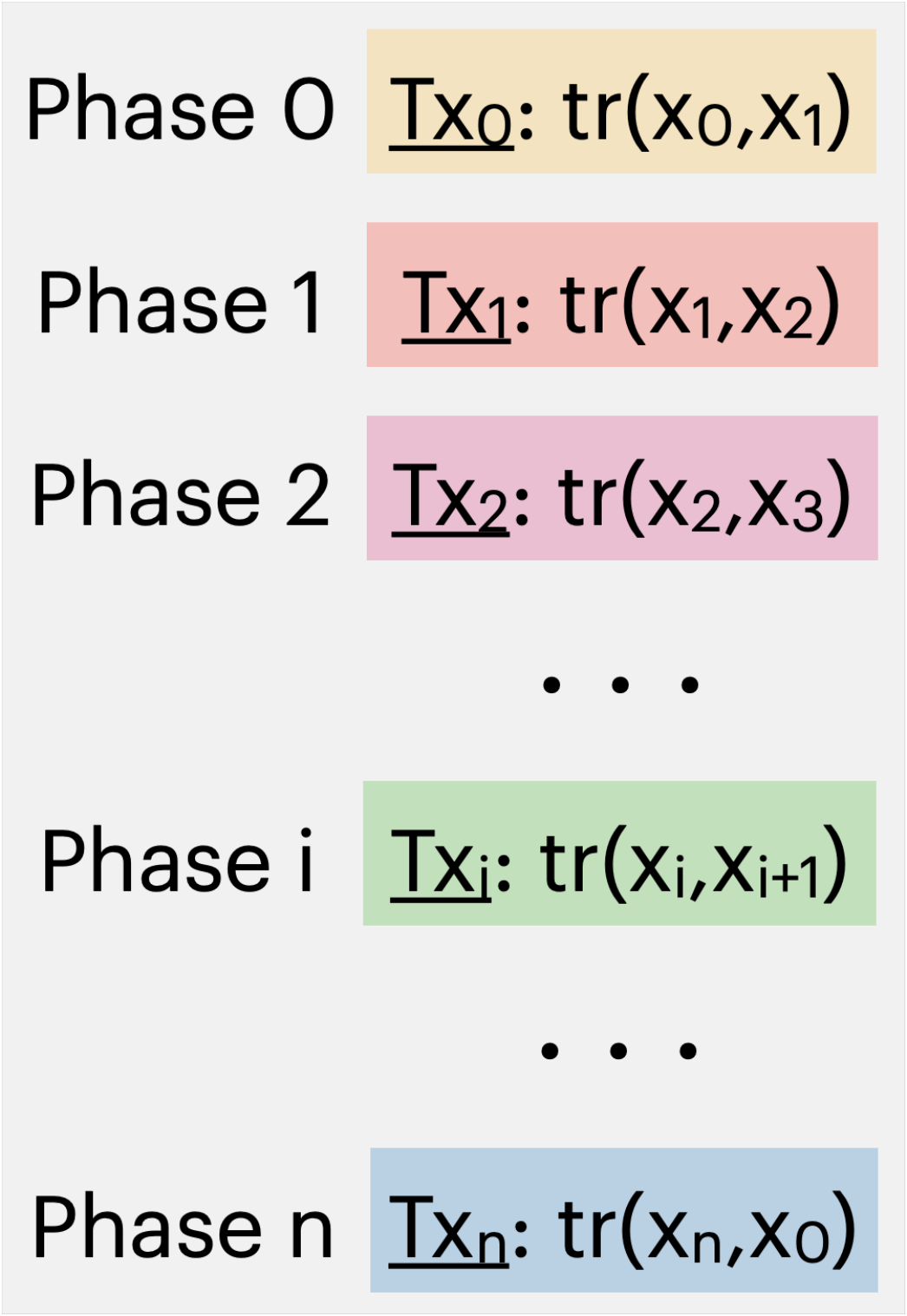}
            \caption{Consensus \\ induced scheduling}
            \label{fig:gap:min}
        \end{subfigure}
        \begin{subfigure}{0.67\linewidth}
            \centering
            \includegraphics[width=\linewidth]{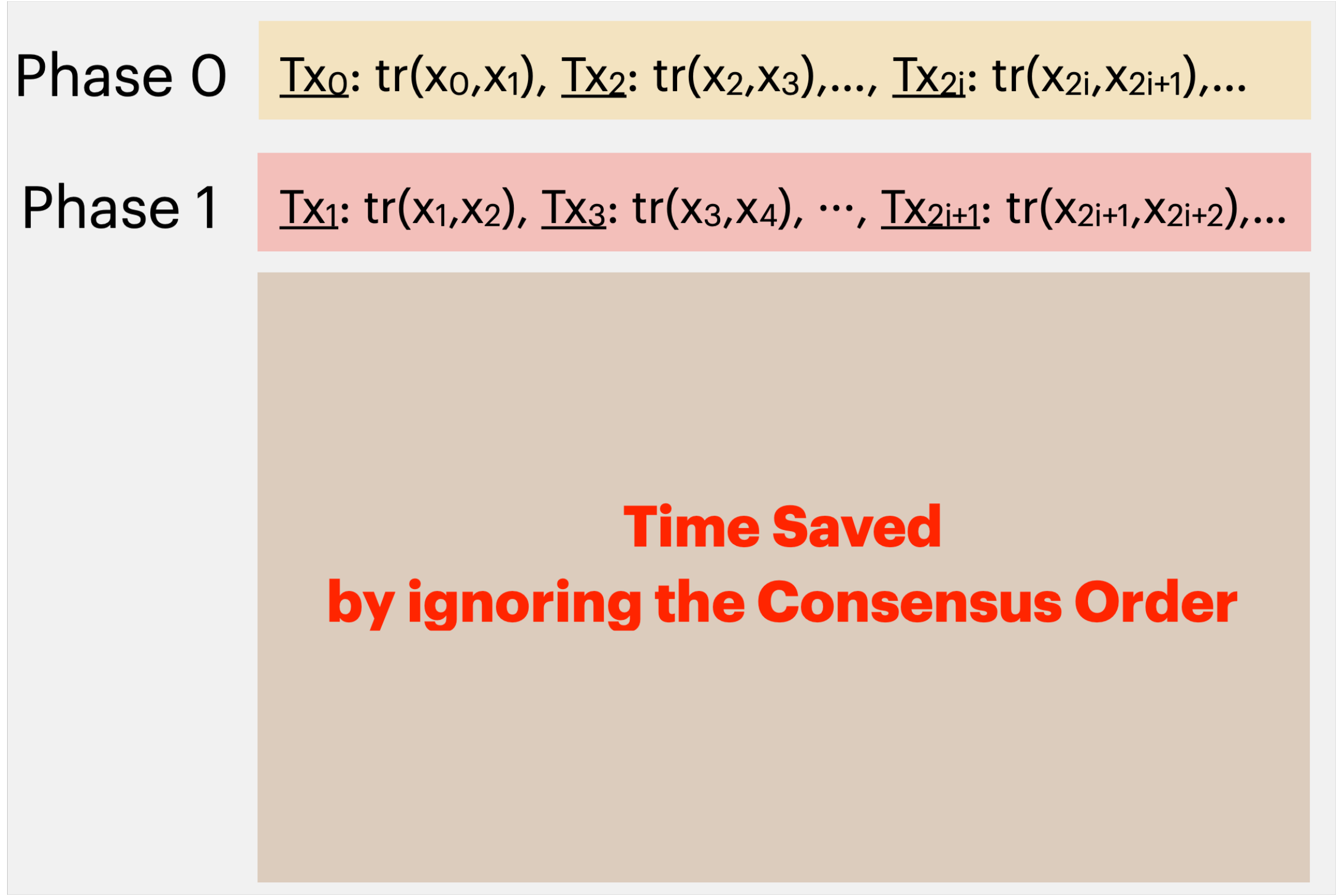}
            \caption{Maximal concurrency\\\ }
            \label{fig:gap:max}
        \end{subfigure}
        \caption{The gap between scheduling a given block's transactions using consensus induced ordering vs. the maximal potential concurrency -- the former takes a linear number of execution steps while the latter only~two.}
        \label{fig:gap}
    \end{figure}

    In this work, we launch a formal study of this gap in order to better understand it and what can be done to narrow it.
    It
    is well known that dependencies between conflicting transactions can be expressed using a conflict graph.
    Hence, we start by introducing a formal framework that enables us to derive transaction scheduling decisions, or a \emph{schedule}, from an undirected conflict graph.
    We prove that such schedules are deterministically serializable, and define latency as well as optimality based on the graph's depth.
    With these, we show an equivalence between finding optimal schedules and minimal graph coloring of the conflict graph when all transactions have similar execution times, which we refer to as \emph{homogeneous transactions}.
    On the positive side, this means that minimal coloring can be used to generate optimal schedules for homogeneous transactions.
    \glsunset{NPH}
    On the other hand, this implies that the problem is \gls{NPH}~\cite{Karp1972}.
    As for \emph{heterogeneous transactions}, we prove that minimal coloring does not guarantee optimal schedules.
    Even
    the generalized weighted minimal coloring problem does not help here.
    Rather,
    we propose a few practical approaches whose exploration is left for future work.

    \paragraph*{Summary of Contributions}
    We
    explore the possibility of generating highly concurrent scheduling algorithms for executing transactions, or smart contracts, that are packaged within the same block of a blockchain system.
    Such schedules must ensure correct executions of all transactions, in particular, in the face of conflicts.
    Such schedules must also be deterministic, as all replicas of the system, validators or miners, must be able to generate and follow them locally and reach the same results everywhere.
    More~specifically:
    \begin{enumerate}[leftmargin=*]
        \item We expose an inherent performance vulnerability that blocks' creators can exploit against existing solutions for concurrent smart contract executions.
        To our knowledge, we are the first to report this.
        \item To alleviate this vulnerability, we propose to break away from the existing approach of relying on the consensus total order, and consider alternative deterministic serializable partial orderings for intra-block transactions.
        \item We establish a formal link between concurrent smart contract execution and deterministic scheduling.
        \item We present a generic formal framework for deriving scheduling decisions, or a schedule, from an undirected conflict graph, and introduce the concept of \emph{sequentially deterministic} schedules.
        This facilitates the exploration of scheduling algorithms and proving their correctness,
        as such proofs only need to show that the resulting schedules are sequentially deterministic.
        \item We prove an equivalence between finding optimal schedules and minimal graph coloring of the conflict graph when all transactions have similar execution times, referred to as \emph{homogeneous} transactions.
        \item For heterogeneous transactions, we exemplify that neither \ coloring nor\ weighted coloring ensures optimal~schedules.
    \end{enumerate}

    \noindent We call this new paradigm \emph{Batch-Schedule-Execute} (BSX) to succeed the XO pattern for it three steps: \emph{Batching} transactions into unordered blocks, creating an efficient deterministic \emph{Schedule} and then executing accordingly.
    We suggest an optional \emph{Order} step (BSXO) to signify that a unique order can be derived form the schedule when applicable.

    \paragraph*{Paper Roadmap} 
    The background and formal preliminaries for this work are presented in \autoref{chap:motivation}.
    We describe the performance vulnerability in \autoref{sec:vulnerability}.
    The generic formal framework is introduced in \autoref{chap:deterministic-scheduling}.
    We introduce the concept of graph scheduling in \autoref{sec:graphing}, and
    discuss its optimality goals in \autoref{chap:optimal}.
    We explore using graph coloring for finding an optimal schedule when transactions are homogeneous in \autoref{chap:homo}.
    Heterogeneity is considered in \autoref{chap:hetero}.
    Related works are surveyed in \autoref{chap:related}, and we conclude with a discussion in \autoref{chap:discussions}.

    \clearpage
    \section{Background}
    \label{chap:motivation}

    A blockchain system is composed of \emph{client nodes}, \emph{validator nodes}, and \emph{observers}.
    In the common Blockchain model, clients submit transactions to validators for execution.
    Validators are responsible for both maintaining the ledger of transactions and executing clients' transactions w.r.t. to the ledger's state.
    All validators must agree on the order in which transactions appear in the ledger, and (logically) execute transactions in that order, thereby realizing a \emph{replicated state machine} semantics~\cite{Schneider1990}.
    Further, consecutive transactions in the ledger are grouped into blocks.
    Observers can access the ledger and read it.

    In this paper, we focus on blockchain systems in which validators continuously agree on the next block of transactions to be executed, and then each executes it locally.
    This can be done in a pipeline, as it is not required to wait until the local execution of a block completes before agreeing on the next block.
    Given the definition of concurrent transactions~\cite{DB-textbook}, all transactions in the same block are
    concurrent~\cite{ParBlockchain, BlockSTM}.
    Therefore, in principle, they can be executed in any order, provided that all validators execute them in the same logical~order.

    \paragraph*{Concurrency} We define an \emph{execution} of a blockchain system to consist of all actions taken by its clients and its validators.
    Our definiton of execution is similar to the definition of a ``schedule'' in Database Theory \cite{DB-textbook}. This is because we define the term ``schedule'' in \autoref{sec:graphing} with semantics that are typical in traditional computer science theory.
    We also use the traditional definitions for execution equivalence~\cite{DB-textbook}, \gls{serializability}~\cite{DB-textbook}, and \emph{strict \gls{serializability}}~\cite{Linearizability}.
    Recall that executions are equivalent if they have the same results for every initial global state. We denote this using~$\sigma \equiv \sigma'$.
    \emph{Strict \Gls{serializability}} also enforces the respective partial real-time order of transaction in addition to \gls{serializability}.
    We envision that especially in consortium blockchains, the number of hardware threads, e.g., using upcoming AMD EPYC ``Turin'' (5gen) CPUs~\cite{AMD-EPYC}, is likely to be larger than the typical number of transactions per block.
    Hence, we focus on the maximum level of parallelism attainable.

        \subsection{State Machine Replication} \label{sec:smr}

    Blockchains are often viewed as a specific instance of \glsxtrfull{SMR}~\cite{Schneider1990}.
    Specifically, we assume a service replicating a \emph{state}, which is composed of a collection of \emph{objects} associated with \emph{operations} or \emph{transactions} that can be invoked on them by clients of the system.
    Transactions may change the content of objects, and may return a \emph{result}, but their execution must be \textbf{deterministic} and can only be affected by the value of the current state and the values of their input parameters.

    In principle, a replicated state machine is assumed to operate normally despite potential replica failures, whose exact types are out of the scope; our work is indifferent to them.
    In active replication, transactions are executed in \emph{all} replicas according to a consistent total order of all transactions.
    The required consistent total ordering is often obtained by repeated execution of some consensus protocol~\cite{Paxos, RAFT, PBFT, tendermint};\ see \autoref{fig:asmr}.

    \begin{figure}
        \centering
        \includegraphics[width=0.7\linewidth]{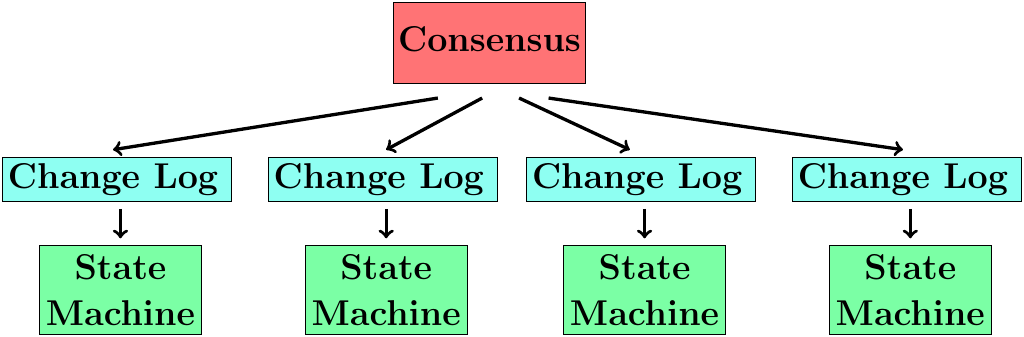}
        \caption{An illustration of Active Replication. }
        \label{fig:asmr}
    \end{figure}

    \subsection{Transaction Conflicts}

    \paragraph{Conflict Graph} Intuitively, a pair of transactions may conflict if each invokes at least one operation on a shared data object that clashes with the corresponding operation of the other transaction~\cite{DB-textbook}\ (see the exact formal definition for read and write operations in \autoref{sec:def-rwsets-conflicts}).\ 
    Given a set of transactions, it is possible to create a \emph{conflict graph} whose vertices are the transactions and there is an (undirected) edge connecting every pair of conflicting~transactions.

    \paragraph{Total Order Tie Breaker} To enable concurrent, yet \gls{serializable}, execution of conflicting transactions and maintain the same logical ordering, previous works~\cite{KD04, ParBlockchain, COS, BlockSTM} used the total ordering order to turn the conflict graph into a \gls{DAG} by directing its edges according to the total ordering.
    The resulting dependency graph is used to ensure that any transaction $tx$ is logically executed before $tx'$ whenever $tx$ is (totally) ordered before~$tx'$.

    In contrast, in this work, we study other methods for directing the conflict graph edges that yield optimal, or near-optimal, concurrency and latency.
    In particular, we explore the feasibility of applying minimal coloring to the conflict graph and directing its edges according to the respective colors associated with the transactions. 
    That is, given two conflicting transactions $tx$ and $tx'$ such that \texttt{color}($tx$) $<$ \texttt{color}($tx'$), we direct the corresponding edge in the conflict graph such that it points to $tx'$.

     \subsection{Transaction Conflicts and Read/Write Sets} \label{sec:def-rwsets-conflicts}

    \paragraph{Read-Write Sets} Each transaction $tx$ is associated with a set of objects it accesses in the global state.
    For simplicity, we assume that transactions may invoke \textsc{read} and \textsc{write} operations on these objects.
    Transaction
    $tx$ is assumed to be associated\footnotemark{} with the sets\ of objects in the global state it \emph{may} read from or write to.
    \footnotetext{
        As in~\cite{KD04,ParBlockchain,COS}, we assume that read sets and write sets are known, e.g., due to explicit annotations, static analysis, or speculative~execution.
    }
    These are denoted $\readset{tx}$ (read-set) and $\writeset{tx}$ (write-set).
    We note that the read-set must include all objects that may be read from in \textbf{any} possible execution.
    This requirement is useful to guarantee the safety in \autoref{thm:graph-sequentially-deterministic}.
    This applies to the write-set as well.

    \newcommand{\true}{\texttt{true}}
    \newcommand{\false}{\texttt{false}}

    \begin{definition}[Conflicts \smaller{[for Read and Write Operations]}]
        Given two distinct transactions $tx_1$ and $tx_2$, we say that $tx_1$ and $tx_2$ are \emph{conflicting} if any of the following holds:
        ($i$) $\readset{tx_1} \cap \writeset{tx_2} \neq \emptyset$ (read-write conflict),
        ($ii$) $\writeset{tx_1} \cap \readset{tx_2} \neq \emptyset$ (write-read conflict),
        ($iii$) $\writeset{tx_1} \cap \writeset{tx_2} \neq \emptyset$ (write-write conflict).
        We denote $tx_1 \conflicts tx_2$ if and only if $tx_1$ and $tx_2$ are conflicting.
    \end{definition}

    \begin{observation}
        Concurrent executions of a set of non-conflicting transactions always end with the same results.
    \end{observation}

    \section{The Performance Vulnerability}
    \label{sec:vulnerability}
    \autoref{fig:gap:min}
    shows an example of a potential block of transactions, in which transaction $Tx_i$ accesses objects $x_i$ and $x_{i+1}$.
    Clearly, in this case, every pair of transactions $Tx_i$ and $Tx_{i+1}$ conflict.
    When using the total ordering induced by their indexes, $Tx_i$ must be executed before $Tx_{i+1}$; this leads to a fully sequential execution.

   On the other hand, it is easy to verify that it is possible to split the set of transactions into two subsets: all even numbered transactions vs. all odd transactions.
   We can then execute all even transactions concurrently followed by a concurrent execution of all odd transactions (assuming enough cores are available).
   That is, we found a schedule that can execute all transactions in two \emph{phases} (or steps).

    A natural question is how typical the example of \autoref{fig:gap} is.
    Here, we separate between what a malicious block creator can do and what happens to an honest but uninformed block creator, whenever total ordering is used to turn conflicts into~dependencies.
    To create maximal damage, a malicious attacker can look for the longest conflict chain, or path, in the conflict graph, and then ensure that all transactions belonging to this path appear consecutively in the block total order.

    Yet, even when an honest, but uninformed, block creator places transactions in an arbitrary order, this could lead to significant performance loss, when the intra-block total order is preserved.
    At the theoretical level, in~\appref{bound-proof}, we show the following (the proof is long, technical, and requires its own non-trivial formal machinery):
    Assuming the length of the longest conflict chain is non-negligible, there is a non-negligible chance that a random total ordering (chosen by an honest node) of these transactions would imply significantly sub-optimal concurrency level.
    We also report in \appref{attack} on an empirical evaluation that supports this theoretical result.
    Altogether, this motivates exploring how we can obtain optimal parallelization using other methods. 

    In more details, as we show in this work, the existence of a simple path of $k$ vertices in the conflict graph corresponds to total orders that imply a minimal latency of at least $k$ (for homogeneous blocks).
    We further prove in \autoref{chap:homo} below that the chromatic number of a conflict graph is the lowest latency attainable for the respective set of transactions and~conflicts.

    The combination of these two facts is used to evaluate the ratio between the worst latency and the best latency for a given block.
    That is, assuming $\ell$ is a lower bound for the longest simple path in a given conflict graph and $\chi$ is an upper bound for its chromatic number, we can use $\frac{\ell+1}{\chi}$ as a lower bound on the ratio between the worst and best possible latencies.

    In the experiment we report in \appref{attack}, we randomly draw a large number of $G(n,p)$ graphs, where $n$ is the number of nodes (or transactions) and $p$ is the probability of having an edge between any two nodes (or a conflict between any two transactions).
    For each such graph, we compute the respective $\ell$ and $\chi$ and compute average $\frac{\ell+1}{\chi}$ values for each pair of $n$ and $p$ values.
    For the vast majority of cases, this yields a non-negligible slowdown.
    \clearpage

    \section{Generic Block \glsfmtshort{ASMR} Framework}
    \label{chap:deterministic-scheduling}

    We introduce a generic framework, called \glsxtrshort{scheduler} Framework, for executing blocks of transactions concurrently in a sequentially deterministic manner and for reasoning about their correctness.
    Particularly, we prove conditions that \glsxtrshort{scheduler}[s] must meet to ensure that the resulting framework is indeed deterministically serializable.
    The blocks
    are executed according to their order, i.e., the execution of transactions must start after the execution of all transactions from previous blocks ended.
    In subsequent sections, we turn our attention to specific \glsxtrshort{scheduler}[s].

    Our framework follows traditional \gls{SMR} techniques by building the transaction history using a multi-consensus protocol~\cite{Paxos}.
    We assume that the consensus protocol creates a reliable stream of blocks, each of which consists of a batch of transactions~\cite{FR97}.
    In our framework, blocks' transactions are executed using a swappable \gls{scheduler},
    that makes scheduling decisions and executes the transactions accordingly.
    The formal definition for a \gls{scheduler} is in \appref{asmr-framework}.

    \autoref{fig:main-loop-flowchart}
    details the main loop for an \gls{ASMR} using a provided \glsxtrlong{scheduler} $\A$ for executing blocks of transactions.
    \begin{figure}
    
        \centering
        \includegraphics
        [width=0.8\linewidth]
        {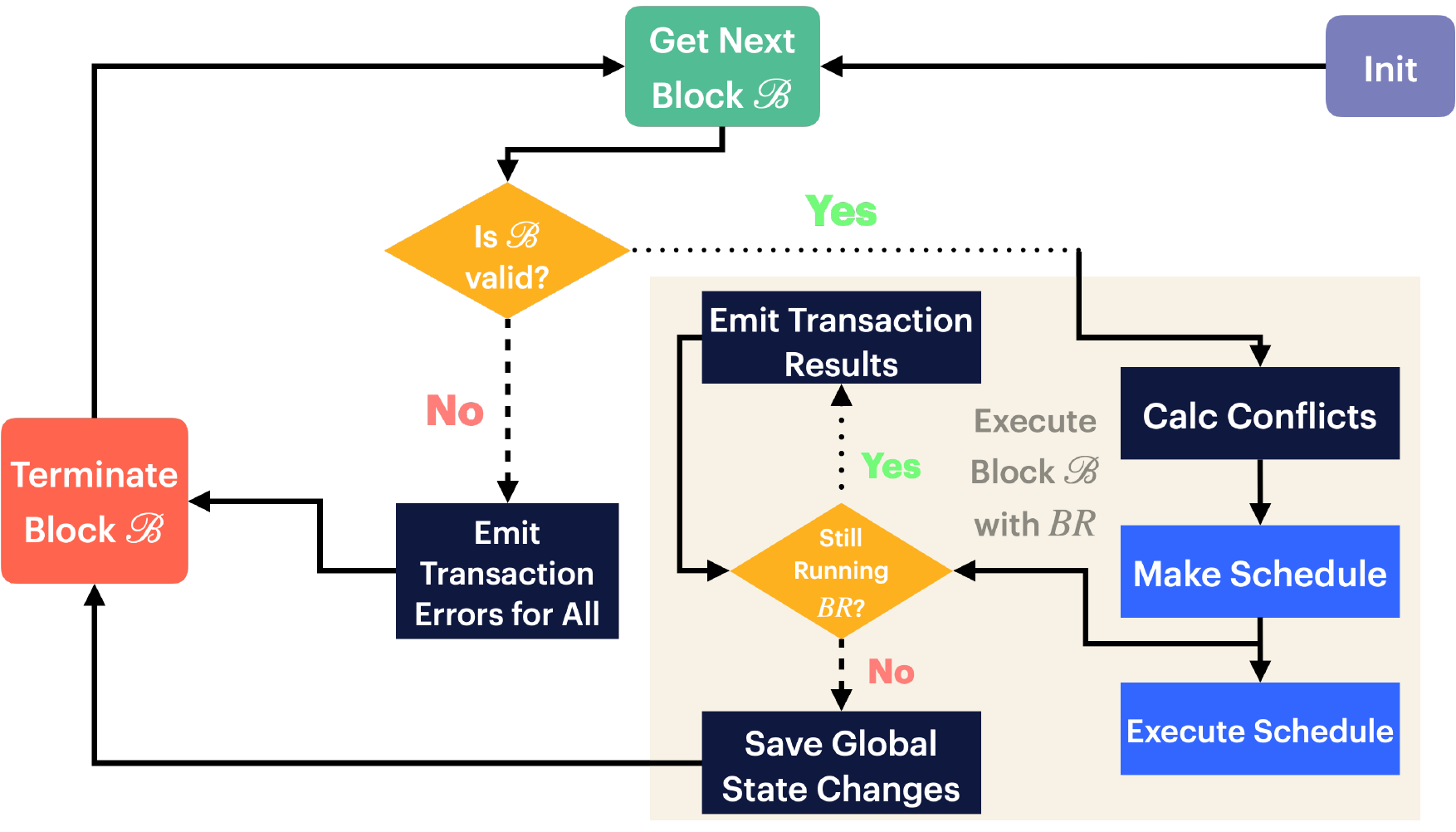}
        \caption{A Flow-Chart of the main loop. }
        \label{fig:main-loop-flowchart}
    \end{figure}
    Intuitively, the main loop of the \gls{ASMR} algorithm works as follows: The next block that was decided on by the consensus protocol is fetched and validated.
    The validation includes checking the sequence number, the hash of the previous block, and similar syntactic issues are verified.
    Then begins the processing of this specific block: the constraints are prepared for the execution of all its transactions.
    For classic total ordering semantics, this would be the process of retrieving the sequential order of transactions in the block.
    Alternatively, this could be the process of computing which transactions are in conflict with other transactions.
    Next, the \gls{scheduler} $\A$ creates a deterministic \gls{schedule}.
    Transactions are executed according to the \gls{schedule}, while the results of the finished transactions are emitted.
    Finally, all global state changes are committed, resulting with the newest version of the global state before continuing to the next~block.
     An execution of \gls{scheduler} $\A$ for some block $\Block$ is any of the executions of these steps for the specific block
    $\Block$ and given some~state.

    \begin{definition}[Sequentially Deterministic \gls{scheduler}]
        We say that a \gls{scheduler} $\A$ is \emph{deterministic} if for every block $\Block$ and for any two executions $\sigma_\Block$ and $\sigma'_\Block$ of $\A$ for block $\Block$, the equivalence $\sigma_\Block \equiv \sigma'_\Block$ is valid.
        Furthermore, $\A$ is \emph{sequentially deterministic} if also for every block $\Block$ there exists a serial execution $\tau_\Block$ such that for every execution $\sigma_\Block$ of $\A$ for $\Block$, $\sigma_\Block \equiv \tau_\Block$~holds.
    \end{definition}
    In \appref{asmr-framework}, we prove that this \gls{ASMR} framework is strictly serializable if and only if the \gls{scheduler} $\A$ is sequentially deterministic.
    In the following sections, we consider specific \gls{scheduler}[s], and prove that they are in fact sequentially deterministic.

    \clearpage
    \section{Graph Scheduling}\label{sec:graphing}

    \subsection{Graph Schedules} \label{sec:schedule}

    \begin{figure}
        \begin{subfigure}{0.3\linewidth}
            \centering
            \includegraphics[width=0.95\linewidth]{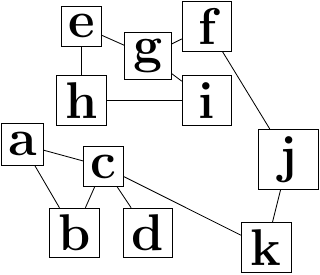}
            \caption{Conflicts Graph}
            \label{fig:conflict-graph-example-def}
        \end{subfigure}
        \hfill
        \begin{subfigure}{0.3\linewidth}
            \centering
            \includegraphics[width=\linewidth]{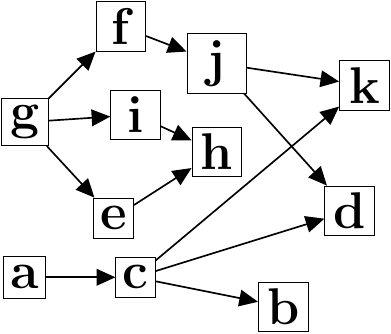}
            \caption{valid schedule
            }
        \end{subfigure}
        \hfill
        \begin{subfigure}{0.3\linewidth}
            \centering
            \includegraphics[width=\linewidth]{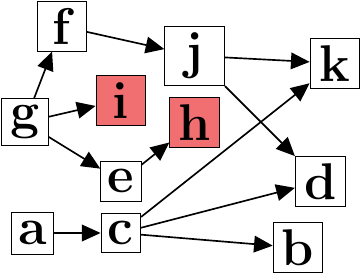}
            \caption{invalid schedule
            }
        \end{subfigure}
        \caption{Conflicts Graph and Schedulers.}
    \end{figure}

    In a similar fashion to previous sections, we assume that for each block $\Block$, $\T$ represents the set of transactions and that conflicts are known for every pair of transactions in $\T$, i.e., the relation $\conflicts$ is known\footnotemark{}.
    \footnotetext{
        Conflicts can be computed directly from the read-sets and write-sets.
        As noted in previous sections, we assume that the read-sets and write-sets are known due to explicit annotations, static analysis, or speculative execution.
    }
    Using the conflict relation of the block $\Block=\left(\T, \conflicts\right)$, we construct its \emph{conflict graph} $\confGraph{\Block} \triangleq \left(\T, \conflicts\right)$, an undirected graph, such that its vertices are the transactions in $\T$ and its undirected edges represent the conflict relation~$\conflicts$.

    \begin{definition}[\Glsxtrshort{schedule} and \Glsxtrshort{scheduling graph}]
        \label{def:schedule}
        Given a subset $\Schd \subseteq \T \times \T$, let $G$ be the directed graph $G = \left(\T, \Schd\right)$.
        If $G$ is a \gls{DAG}, we denote $\SchG{\Schd} \triangleq G$; and say that $\GSchd$ is the \emph{\gls{scheduling graph}} induced by $\Schd$ and that $\Schd$ is the corresponding \emph{graph-schedule (\gls{schedule})}.
    \end{definition}

    \begin{definition}[Valid \Glsxtrshort{schedule}]
        A given \gls{schedule} $\Schd$ is said to be \emph{valid} if the corresponding $\GSchd$ satisfies the following property:
    For all pairs of transactions $tx, tx' \in \T$ such that $tx \conflicts tx'$,  $\GSchd$ contains a directed path either from $tx$ to $tx'$ or vice versa.
    \end{definition}

Hereafter, we only refer to valid schedules unless specifically noted otherwise and drop the word ``valid''.

    \subsection{Scheduling-Graph Oriented \glsfmtshort{scheduler}} \label{sec:graphed-schdlr}

    In this section, we present a \gls{scheduler} that is based directly on the conflict graph.
    It requires explicit synchronization, but only between conflicting transactions.
    The type of synchronization required can be implemented, for example, using an efficient Go-like signaling mechanism \cite{golang}, as elaborated below.
    We call this \gls{scheduler} the \emph{Scheduling Graph \gls{scheduler}}, presented in \autoref{alg:graphschdr}, and denote it by~$\GRAPHA$.
    The method used to create new schedules is intentionally not implemented, since $\GRAPHA$ is presented as a generic \gls{scheduler} for all valid graph schedules.
    The correctness of the \emph{Graph \gls{scheduler}} relies almost entirely on the validity of the graph schedule, as it is used to prevent conflicting transactions from running concurrently.
    In \autoref{sec:greedy} and in subsequent sections, we discuss particular techniques for creating graph schedules that can be seen as specific extensions of $\GRAPHA$.

    \begin{algorithm}[]
        \caption{\Gls{scheduling graph} \gls{scheduler} $\GRAPHA$}
        \label{alg:graphschdretc}
        \footnotesize
    \begin{algorithmic}[1]
        
       \Function{CreateExecutionDS}{schedule $\Schd$, global-state s}
            \State \textbf{add} dummy transactions $start, end \notin \T$ and the edges \Statex
                \begingroup \hspace{2em} \tiny
                $\tiny\left\{\left(start, tx\right) \mid \indeg[\GSchd]{tx} = 0\right\}$ $ \cup$ $\tiny\left\{\left(tx, start\right) \mid \outdeg[\GSchd]{tx} = 0\right\}$
                \endgroup \hspace{2em}
            \ForAll{$\left(tx_i, tx_j\right) \in \Schd$}
                \State \textbf{create} new signal primitive $\texttt{signal}\!\left(tx_i,tx_j\right)$ with locked state
            \EndFor
        \EndFunctionExplicit
        \Function{$\GRAPHA\!::$init-execution}{schedule $\Schd$, global-state s}
            \State \texttt{signals}, $start$, $end$ $\gets$ \Call{CreateExecutionDS}{$\Schd$, s}
            \ForAll{$tx\in\T$}
                \Spawn[($tx$)]
                    \State \textbf{wait until all} $\texttt{signal}\!\left(*, tx\right)$ are unlocked \label{alg:graphed:execute-schedule:wait-locked}
                    \vspace{0.5em}
                    \State \textbf{serve} $\readset{tx}$ by taking the latest version of each object 
                    \State \textbf{execute} $tx$
                    \State \textbf{store} new values of $\writeset{tx}$ \vspace{0.5em}
                    \State \textbf{emit} transaction results for $tx$
                    \State \textbf{unlock all} $\texttt{signal}\!\left(tx, *\right)$ \label{alg:graphed:execute-schedule:unlock-signals}
                    \EndSpawn
            \EndFor 

            \State \Return execution $\gets$ (\texttt{signals}, $start$, $end$)
        \EndFunctionExplicit
    \algstore{alg4}
    \end{algorithmic}
        \label{alg:graphschdr}
        \label{alg:graphed:execute-schedule}
    \begin{algorithmic}[1]
        \algrestore{alg4}
        \Procedure{$\GRAPHA\!::$start-execution}{execution e}
            \State \textbf{unlock all} $\texttt{signal}\!\left(start, *\right)$
        \EndProcedure

        \Function{$\GRAPHA\!::$is-execution-running}{execution e}
            \If{all $\texttt{signal}\!\left(*, end\right)$ are unlocked} \Return true
            \Else \ \Return false
            \EndIf
        \EndFunction

        \Function{$\GRAPHA\!::$next-execution-results}{execution e}
            \State \Return transaction results that have been emitted
        \EndFunction

        \Function{$\GRAPHA\!::$state-changes}{execution e}
            \State \Return state changes 
        \EndFunction
    \end{algorithmic}
    \end{algorithm}

    \begin{theorem} \label{thm:graph-sequentially-deterministic}
    \Gls{scheduler} $\GRAPHA$ is sequentially deterministic, assuming only valid schedules are used as input.
    \end{theorem}

    \begin{proof}
        We prove this by showing that every execution of $\GRAPHA$ is serializable, and that it is deterministic.
        Fix some block $\Block$ for which we will show the necessary properties.

            {(\texttt{serializable})} We start by proving that each execution of $\GRAPHA$ for the given block $\Block$ is \gls{serializable}.
        Consider some execution $\sigma$ of $\GRAPHA$ for $\Block$.
        Let $\GSchd$ be the \gls{scheduling graph} of the schedule for block $\Block$ induced by the \emph{valid} \gls{schedule} $\Schd$.
        Let $\tau$ be a serial execution obtained by topologically sorting the transactions in the \gls{DAG} $\GSchd$.
        Given two transactions $tx \conflicts tx'$, we know that by definition \gls{schedule} $\Schd$ contains a path between them. \gls{wlog} $tx \longhashpath_{\Schd} tx'$ with the path $tx=tx_0,tx_1, \ldots, tx_n=tx'$.
        Consider now how $\GRAPHA$ executes given $\Schd$.
        For every $0 < i \leq n$ the thread for $tx_i$ waits until all signals $\texttt{signal}\!\left(*, tx_i\right)$ are unlocked before it starts executing (line~\ref{alg:graphed:execute-schedule:wait-locked}).
        In particular, $tx_i$ must wait until the signal $\texttt{signal}\!\left(tx_{i-1}, tx_i\right)$ is unlocked by the thread of $tx_{i-1}$.
        This happens when the thread of $tx_{i-1}$ reaches line~\ref{alg:graphed:execute-schedule:unlock-signals}, that is, only after $tx_{i-1}$ finishes executing.
        We can apply this to each of the edges of the path and deduce that, in $\sigma$, $tx=tx_0$ must finish executing before $tx'=tx_n$ starts.
        
        Also, $tx$ must be ordered before $tx'$ in $\tau$, since any valid topological order must respect the transitive dependency between these two transactions.
        Therefore, all conflicts in $\Block$ are ordered in the same way in both $\tau$ and $\sigma$.
        Moreover, for any two non-conflicting transactions, the values read and written by such transactions are independent of their relative order, and in particular, are the same in $\tau$ and in $\sigma$.
        In summary, $\tau$ is a valid sequential execution that is conflict-equivalent to $\sigma$, and thus $\sigma$ is \gls{serializable}~\cite{DB-textbook}.

            {(\texttt{deterministic})} Next, we show that $\GRAPHA$ is deterministic.
        Consider two executions $\sigma$ and $\sigma'$ of $\GRAPHA$ for $\Block$.
        Since we assume that the creation of the schedule is deterministic, both instances of \autoref{alg:graphschdr} use the same \gls{scheduling graph}.
        Also, for arguments identical to those above, there exist two serial \gls{schedule}[s] $\tau \equiv \sigma$ and $\tau' \equiv \sigma'$ that were obtained by topologically sorting $\GSchd$.

        Now, we know that two arbitrary conflicting transactions are ordered in the same way in both $\tau$ and $\tau'$. This is guaranteed by definition because they are both topological sorts of the same \gls{scheduling graph}.
        Therefore, $\tau \equiv \tau'$ since they are conflict equivalent; thus $\sigma \equiv \tau \equiv \tau' \equiv \sigma'$.
        In summary, we have $\sigma \equiv \sigma'$ as needed.
    \end{proof}
    \subsection{Latency} \label{sec:graph-latency}
    
    We assume that $\T$ is accompanied by a mapping function $\len:\T \longrightarrow \symb{N+}$ that assigns an \emph{execution duration}, or \emph{length}, for each transaction.
    The length of each transaction need not be the exact execution time, but rather an abstract sense of time that allows for distinguishing how much longer one takes compared to the other.

    \begin{definition}[Latency]
        The latency of a \gls{schedule} $\Schd$ is the weight\-ed depth of $\GSchd$, that is, the maximum weighted length\footnote{When the path length is measured by vertices rather than~edges.} of any simple path in $\GSchd$, denoted:
         \begin{equation}
         \footnotesize
            \Lt[\len]{\Schd} \triangleq
            \Depth[\len]{\GSchd} =
            \max \left\{ \lenOf{P}  \mid P \text{ simple path in }\Schd \right\}.
        \end{equation}
        When the weighted length of a path $P$ is the sum of the length of its vertices, $\lenOf{P} \triangleq \sum_{v \in P} \lenOf{v}$.
    \end{definition}

    \autoref{thm:latency-preservation} explains why, in a sense, the latency of a \gls{schedule} is the time it takes to execute it.
    This explains the motivation behind the above definitions.
    
    \begin{theorem}[Latency Preservation of $\GRAPHA$]
        \label{thm:latency-preservation}
        Let $\Schd$ be an optional \gls{schedule} of $\Block$.
        Assume that \autoref{alg:graphed:execute-schedule} is executed in an environment such that:
        \begin{enumerate*}
            [label={\smaller(}\roman*{\smaller)}]
            \item the number of processors is unbounded,
            \item the delay for passing synchronization signals is negligible,
            \item the cost of executing an operation on a global object has no additional synchronization overhead compared to a single-threaded implementation, and
            \item $\lenOf{tx}$ is the time duration it takes to execute the transaction $tx$.
        \end{enumerate*}
        Then, the time it takes to execute \autoref{alg:graphed:execute-schedule} for a given \gls{schedule} $\Schd$ and the block of $\Block$ is exactly $\Lt[\ell]{\Schd}$.
    \end{theorem}
    \begin{proof}
        Since the number of processors is unbounded, each transaction may be executed without interruptions as soon as it receives all of its signals.
        No thread of any transaction waits unless it has not received one of the signals it waits on.
        Thus, the total execution time is the same as the weighted depth of $\GSchd$, which is equal to $\Lt[\ell]{\Schd}$ by definition.
    \end{proof}
    \subsection{Level-Induced Graph Scheduling} \label{sec:greedy}
    We now present an extension to $\GRAPHA$ \ \gls{scheduler} that creates a schedule from any complete partition of the transactions to \gls{conflict-free}.
    We call these partitions \emph{legal}.
    We call it the \emph{Level(-induced) Graph \gls{scheduler}} (\autoref{alg:greedy-schedule}) and denote it by $\GREEDYA$.
    \autoref{alg:greedy-schedule} describes the new method of creating valid schedules using \gls{Greedy Schedule}.
    \begin{algorithm}
        \footnotesize
\begin{algorithmic}[1]

        \Function{LevelSchedule}{Ordered Partition $\bigsqcup_{i=1}^{k-1} B_i = \T$}
            \Statex \hspace{1em}\ \ \textbf{Require:} $B_i$ is \glsxtrshort{conflict-free} i.e. $\forall tx,tx'\in B_i \: : \: tx \notconflicts tx'$
            \State $\GSchd = (T, \Schd) \gets \left(\emptyset, \emptyset\right)$
            \State $B_0 \gets \emptyset$
            \For{$B_i = B_1,\ldots, B_k$ ($i$ increases)}
                \State $T \gets T \cup B_i$
                \For{$B_j = B_{i-1},\ldots, B_0$ ($j$ decreases)}
                    \State $E \gets \left\{\left(tx_j,tx_i\right) \in B_j \times B_i \mid  tx_j \conflicts tx_i\right\}$ \label{alg:greedy-schedule:iter_start} \label{alg:greedy-schedule:confpairs}
                    \State $P \gets \left\{\left(tx,tx'\right) \in \T \times \T \mid  \GSchd \text{ contains a path } tx \rightsquiglearrow tx'\right\}$ \label{line:greedy-check}
                    \State $\Schd \gets \Schd \cup \left(E \setminus P\right)$ \label{alg:greedy-schedule:iter_end}
                \EndFor
            \EndFor
            \State \Return $\GSchd$
        \EndFunctionExplicit
        \algstore{alg3}
            \end{algorithmic}
        \caption{Level Graph \gls{scheduler}, $\GREEDYA$}
        \label{alg:greedy-schedule} \label{alg:graphed:make-schedule}
    \begin{algorithmic}[1]
    \algrestore{alg3}
    \Function{$\GREEDYA\!::$make-schedule}{Block $\Block=\left(\T, \conflicts\right)$}
        \State $\bigsqcup_{i=1}^{k-1} B_i \gets$ Obtain a deterministic partition of $\T$ with $\confGraph{\Block}$
        \Statex \hspace{1em}\ \ \textbf{Require:} $B_i$ is \glsxtrshort{conflict-free}, i.e., $\forall tx,tx'\in B_i \: : \: tx \notconflicts tx'$
        \State order $B_1, \ldots, B_k$ in some deterministic way \label{alg:greedy:make-schedule:ordered-part}
        \State $\Schd \gets \Call{LevelSchedule}{\left[B_1, \ldots, B_k\right]}$
        \State \Return $\Schd$
    \EndFunctionExplicit
\end{algorithmic}
    \end{algorithm}
    \begin{lemma} \label{lemma:valid-greedy-schedule}
    When given a legal partition, the \Glsxtrshort{Greedy Schedule} in \autoref{alg:greedy-schedule} returns a valid \gls{schedule}.
    \end{lemma}
    \begin{proof}
        Let two transactions be such that $tx_i \conflicts tx_j$. \Gls{wlog} assume that $tx_i \in B_i \neq B_j \ni tx_j$.
        Since $tx_i \conflicts tx_j$, some iteration of \hyperref[alg:greedy-schedule:iter_start]{lines~\ref*{alg:greedy-schedule:iter_start}}~--~\ref{alg:greedy-schedule:iter_end} checks if $\Schd$ already contains a path between them and otherwise adds a direct edge between them (line~\ref{alg:greedy-schedule:iter_end}).
        Moreover, for two $tx, tx' \in B_i$ we know that $tx \notconflicts tx'$ because $B_i$ is \gls{conflict-free}.
        Thus, for any two transactions $tx_i \conflicts tx_j$ the \gls{schedule} contains a directed path between them, so the \gls{schedule} is valid.
    \end{proof}
    \subsubsection{Completeness of Level-induced Scheduling}
    We use the term \emph{level schedule} to refer to schedules created from some legal partition using the \gls{Greedy Schedule}, and coin \emph{level(-induced) graph scheduling} for the process of creating such schedules.
    In this section, we show that Level Graph Scheduling is complete; that is, there exists an equivalent level schedule for any valid graph schedule that may be executed by \autoref{alg:graphed:execute-schedule} having a latency that is no worse.
    This is the conceptual sense of \autoref{thm:greedy-Completeness} and its proof.

    \begin{lemma}
        For some arbitrary block of transactions $\Block=\left(\T, \conflicts\right)$, denote two arbitrary valid schedules $\Schd$ and $\Schd'$.
        If (\ref{eq:sched-conf-equiv}) below holds, then all executions of \autoref{alg:graphed:execute-schedule} for both $\Schd$ and $\Schd'$ are equivalent.
        \begin{equation} \label{eq:sched-conf-equiv}
        \small 
            \forall \;\; tx_a \conflicts tx_b \ (tx_a, tx_b \in \T) \:: \:tx_a \rightsquiglearrow_\Schd tx_b \ \Longrightarrow \  tx_a \rightsquiglearrow_{\Schd'} tx_b
        \end{equation}
    \end{lemma}
    \begin{proofsketch}
        All executions are conflict equivalent~\cite{DB-textbook}.
        $\square$
    \end{proofsketch}

    \begin{algorithm}
        \footnotesize
        \begin{algorithmic}[1]
            \Function{ConvertToColoring}{$\Schd$}
                \State \textbf{let} $c:V\to\symb{N+}$
                \State $S \gets \left\{v \mid \indeg[\GSchd]{v} = 0\right\}$
                \State $l \gets 0$
                \While{$\Schd \neq \emptyset$} \Comment{Must stop because $\GSchd$ is a \gls{DAG}}
                    \State $l \gets l + 1$ \Comment{After the \textbf{loop} $l$ holds the number of colors}
                    \ForAll{$v\in S$}
                        \State $c\left(v\right) \gets l$
                    \EndFor
                    \State $S \gets \left\{v \mid \exists \left(u, v\right) \in \Schd \land u \in S\right\}$
                        \Comment{Performs BFS\ on $\GSchd$}
                \EndWhile
                \State \Return $c:V\to\left\{1, \ldots, l\right\}$
            \EndFunctionExplicit
        \end{algorithmic}
        \caption{Code for the $\textproc{ConvertToColoring}$ function}
        \label{alg:converttocoloring}
    \end{algorithm}

    \Cref{thm:convert-to-partition} introduces the \textproc{ConvertToColoring} function (\autoref{alg:converttocoloring}). This is part of the theoretical machinery that we need to prove \autoref{thm:greedy-Completeness} and other theorems later~on.
        Its proof appears in \appref{convert-to-partition}.
    \begin{lemma} \label{thm:convert-to-partition}
        Given a valid schedule $\Schd$ for some block $\Block=(\T=V, \conflicts$\-$=E)$, find the partition\footnote{We calculate the partition using a function that assigns numbers, or colors, to each transaction because of the relation to graph coloring. Each number (color) represents one of the disjoint sets of the partition. The relation to graph coloring is further explored in subsequent sections.} $\sqcup_{i=1}^k T_i$ using the $\Call{ConvertToColoring}{\Schd}$ function from \autoref{alg:converttocoloring}.
        As $\Schd$ is valid, $\sqcup_{i=1}^k T_i$ is a legal partition of~$\T$ (using $k$ nonempty~sets).
    \end{lemma}

    Two graph schedulers are \emph{equivalent} if all executions of \autoref{alg:graphed:execute-schedule} are equivalent for both schedules.

    \begin{theorem}[Completeness of Level Scheduling] \label{thm:greedy-Completeness}
        For any valid schedule $\Schd$, there exists a legal partition $\bigsqcup_{i=1}^k B_i = \T$ such that $\Schd$ is equivalent to $\Schd'=\Call{LevelSchedule}{\sqcup_{i=1}^k B_i}$, and (\ref{eq:greedyschedule-lte}) holds.
        \begin{equation} \label{eq:greedyschedule-lte}
        \small 
            \Lt[\ell]{
                \smash{
                \underbrace{
                \textproc{LevelSchedule}\!\left(\sqcup_{i=1}^k B_i\right)
                }_{\Schd'}
                }\vphantom{\textproc{LevelSchedule}\!\left(\sqcup_{i=1}^k B_i\right)}
            } \vphantom{\underbrace{\textproc{LevelSchedule}\!\left(\sqcup_{i=1}^k B_i\right)}_{\Schd'}}
            \leq \Lt[\ell]{\Schd}
        \end{equation}
    \end{theorem}

    For lack of space, the proof is delayed to \appref{greedy-main-proof}.
    
    \clearpage
    \section{Optimal Graph Scheduling and Graph Coloring} \label{chap:optimal}
    In \autoref{chap:deterministic-scheduling} and \autoref{sec:graphing}, we presented a technique for designing concurrent \gls{serializable} \gls{scheduler}[s] \gls{ASMR} without introducing inter-replica state inconsistencies.
    $\GRAPHA$ is a generic \gls{scheduler} for valid graph schedules with a proof of correctness.
    Below,
    we define a formal optimization problem that minimizes latency and analyze it from a theoretical computational perspective.

    \paragraph{Optimality Goals.} There are multiple possible ways to define an optimal \gls{schedule}.
    The definition of an optimal \gls{schedule} depends on the property we wish to optimize.
    For example, we may wish to optimize the execution latency of the entire block or the average latency among the transactions in the block.
    Other options include improving the tail latency, e.g., of the 0.95 quantile execution time, etc.
    These parameters also depend on the environment used for execution, mainly the number of cores available.
    Specifically, a higher number of cores may change the values of these metrics for a given \gls{schedule}.
    Here, we focus on optimizing the total latency of the block assuming an unbounded number of cores is available.

    \subsection{Minimizing Latency}
    Below, we assume that the block is given as the triplet of the transactions, the conflicts and their lengths $\Block = \left(\T, \conflicts, \len\right)$.

    \begin{definition}[Optimal \Glsfmttext{schedule}]
        An \emph{optimal} \gls{schedule} is a valid \gls{schedule} with a minimal latency among all valid schedulers of $\Block$.
        This value is called the \emph{optimal latency} of $\Block$ and denoted as:
        \begin{equation}
            \MinLt[\len]{\Block} \triangleq
            \min_{\Schd}\left(\Lt[\len]{\Schd}\right).
        \end{equation}
    \end{definition}

    \paragraph{Search Problem.} We define the formal search problem called $OptimalSchedule$ such that any optimal \gls{schedule} is a valid solution.
    \begin{equation}
    \footnotesize 
        \symb{OptimalSchedule} \triangleq \biggl\{\left(\Block, \Schd\right) \biggl|\: \begin{array}{cc}\Schd \text{ is a \glsxtrshort{schedule} of } \Block \\ \land\: \Lt[\len]{\Schd} = \MinLt[\len]{\Block}\end{array}\biggl\}
    \end{equation}

    \subsection{NP Hardness} \label{sec:nphardness}
    In~\appref{nphardness}, we prove that determining minimal latency and optimizing it is \gls{NPH} by showing a
    reduction from the Graph Vertex Coloring Problem (GCP) that is known to be \gls{NPH}~\cite{npc-guide}.
    Recall, that in the GCP problem the graph is partitioned into disjoint \glsxtrshort{IS}[s], each representing one ``\emph{color}''.
    An \gls{IS} is a vertex set containing no adjacent vertices~\cite{npc-guide}.
    Thus, by definition we deduce two key observations: one being that a subset of transactions $T\subseteq\T$ is \gls{conflict-free} if and only if $T$ is an \gls{IS} in the conflict graph $\confGraph{\Block}$; and the other is that the legal partitions we referred to in \autoref{sec:greedy} are valid colorings of the conflict~graph.
    The minimum possible number $\chromatic{\cdot}$ of colors needed to color a graph is called its \emph{chromatic number}.

    \clearpage
    \section{Homogeneous Transactions}
    \label{chap:homo}
    \label{chap:equal-lengh}

    In this section, we assume that all transactions have a similar execution time.
    We remove this in \autoref{chap:hetero}.

    \paragraph{Homogeneity in Execution Times} In many practical cases, all transactions in a given block have the same execution duration.
    This is true, e.g., in workloads where all transactions do the same basic operations, e.g., transactions that update a register or transfer an asset from one account to another.
    \begin{definition}[\Glsxtrshort{homogeneous transactions}]
        A block of transactions $\Block=\left(\T, \conflicts\right)$ with a corresponding length function $\len:\T \to \symb{N+}$ is said to be a block of \emph{\gls{homogeneous transactions}} if some $c\in \symb{N+}$ exists such that the following 
        condition
        \ holds : $ 
            \label{eq:homo-condition}
            \forall\, tx \in \T : \lenOf{tx} = c
        $. 
        
        In such a case, we may formally omit the unit length function $\ell$ and use the value $1 \in \symb{N+}$ as a substitute whenever $\lenOf{\cdot}$ was previously used.
        We use the symbol $\mathds{1}$ to represent the length function $\len:\T\to\left\{1\right\}$.
        We may use $\mathds{1}$ as a substitute for the length function $\len$ or to eliminate any doubt that we are referring to a homogeneous block, regardless of the actual value of $c$.
    \end{definition}

    \paragraph{$\epsilon$-Homogeneity} Note that in some cases, we might also consider a block of transactions to be ``almost'' homogeneous if the execution durations of all transactions are fairly similar.
    Precisely, consider a case where time differences may exist s.t. $\forall \; tx_1, tx_2 \in \T :\: \left|\lenOf{tx_1} - \lenOf{tx_2}\right| \leq \epsilon$ for some $\epsilon > 0$.
    If $\epsilon$ is a negligible amount of time, we call such a block \emph{$\epsilon$-homogeneous} and consider it as a homogeneous block.
    The exact details of $\epsilon$ and what makes $\epsilon$ negligible are left out since they usually depend on the type of the workload, the environment used for execution, and other factors related to the specific applications and the use case.

    In \autoref{sec:lower}, we show that using a minimal vertex coloring we can find an optimal \gls{schedule} for \glsxtrshort{homogeneous transactions}.
    In \autoref{sec:color-equiv}, we show that finding an optimal schedule in the homogeneous case is equivalent to vertex~coloring.

    \subsection{Optimal Graph Schedules via Min Coloring} \label{sec:lower}
    We call the \gls{scheduler} that is a combination of $\GREEDYA$ and a minimal coloring algorithm the \emph{Minimal Coloring Graph \gls{scheduler}} ($\MINCOLORA$).
    Under the assumption that transactions are homogeneous and that there are more available cores than the maximal number of transactions that obtain the same color, \autoref{thm:lower-bound} implies that $\GREEDYA$ yields the fastest execution time for a block of transactions when combined with a minimal vertex colorer.
    Obviously, the maximal number of transactions of the same color is bounded by the number of transactions in a block.
    For a consortium blockchain, it is quite reasonable to assume that the validators would run on servers with multiple hundreds of cores~each.

   \begin{algorithm}
        \caption{Pseudocode for $\MINCOLORA$}
        \label{alg:mincolorgschd}
        \footnotesize
\begin{algorithmic}[1]
        \Function{$\MINCOLORA\!::$make-schedule}{Block $\Block=\left(\T, \conflicts\right)$}
            \State $\bigsqcup_{i=1}^{k-1} \T_i \gets$ some minimal coloring of $\confGraph{\Block}$ deterministically
            \State order $\T_1, \ldots, \T_k$ in some deterministic way \label{alg:mincolor:make-schedule:ordered-part}
            \State $\Schd \gets \Call{LevelSchedule}{\left[\T_1, \ldots, \T_k\right]}$
            \State \Return $\Schd$
        \EndFunctionExplicit
\end{algorithmic}
    \end{algorithm}

    \paragraph{Minimal Coloring Level Schedule} 
    Suppose that for a block $\Block$ of \gls{homogeneous transactions} we are given a minimal coloring $c_{\min}:\T\to\left\{1,\ldots, k\right\}$ for the corresponding conflict graph $\confGraph{\Block}$.
    Now, denote the \gls{schedule} $\Schd_{c_{\min}}$ obtained by using the \gls{Greedy Schedule} using the legal partition derived from the coloring.
    For clarity, the ordered partition $\bigsqcup_{i=1}^k T_i = \T$ used by \autoref{alg:greedy-schedule} is constructed from $c$ such that $T_i := \left\{tx\in\T \mid c\left(tx\right) = i\right\}$.
    The following two lemmas\ discuss the properties of $\Schd_{c_{\min}}$.

    \begin{lemma} \label{thm:convert-to-partition-num-colors}
        The number of colors used in the process described in \Cref{thm:convert-to-partition} is exactly $\Depth{\SchG{\Schd}}$.
    \end{lemma}
    \begin{lemma}
        \label{thm:lt-numbercolors}
        $\Lt[\unitlen]{\Schd_{c_{\min}}} = k$.
    \end{lemma}
    \noindent For brevity, the proof is delayed to \appref{homo-proofs}.
    \begin{theorem}
        \label{thm:lower-bound}
        $\Lt[\unitlen]{\Schd_{c_{\min}}} = \MinLt[\unitlen]{\T}$.
    \end{theorem}
    \noindent This proof is also delayed to \appref{homo-proofs}.

    \begin{corollary}[Optimal Homogeneous Schedules using Minimal Coloring] \label{thm:optimal-mcgbr}
        Given some correct minimal vertex coloring algorithm $C$, the resulting schedule of the execution of (\ref{eq:homomincolorsolve}) is an optimal schedule for any homogeneous block $\Block = \left(\T, \conflicts, \len \equiv \unitlen\right)$.
        \begin{equation}
            \Schd = \textproc{LevelSchedule}\Big(C\big(\confGraph{\Block}\big)\Big)
            \label{eq:homomincolorsolve}
        \end{equation}
    \end{corollary}

    \paragraph{Relation to Previous Works.}
    Within the context of the $\GRAPHA$ framework, previous works, such as \cite{KD04, ParBlockchain, COS}, can be viewed as techniques for transforming the conflict graph into a schedule. This is by directing the edge between every pair of conflicting transactions according to their respective order in the total ordering.
    In contrast, we derive the edge-directing decisions in a way that guarantees that we minimize the block latency instead of relying on some predetermined order.
    Rearranging the logical order of transactions within the same block is possible because these transactions correspond to concurrent operations.
    Our formalization captures all possible graph schedulers, and serves as a generic proof platform for such schedulers, while previous works target one specific~scheduler.

    Further, the schedule we generate in \autoref{alg:greedy-schedule} has no redundant edges whenever there is already a dependency path between conflicting transactions (\hyperref[alg:greedy-schedule:iter_start]{lines~\ref*{alg:greedy-schedule:iter_start}}~--~\ref{alg:greedy-schedule:iter_end}).
    Having fewer edges in the scheduler keeps the synchronization logic simpler and reduces the overhead added between the execution of transactions.

    \subsubsection{Improving Tail-Latency} \label{sec:homo-tail-latency}
    \label{sec:homotail-proofs}

    In \autoref{sec:lower} we showed how to minimize the latency of execution of the block of transactions as a unified task.
    We call this latency the \emph{block-latency}.
    However, clients do not necessarily need to wait for the entire block to complete before receiving the results of their transactions.
    In fact, results can be returned to clients as soon as individual transaction executions finish, even before the entire block completes\footnote{This is the behavior of the main loop for \gls{ASMR} in \autoref{fig:main-loop-flowchart}. }.

    Minimizing block-latency is our primary goal; our secondary goal is to improve the latency of the individual transaction latencies, after minimizing the block-latency.
    \autoref{fig:tail-latency-good-vs-bad} depicts the impact the order of transactions in the \gls{schedule} has on the tail latency.
    In both case~(a) and case~(b) the block latency remains the same.
    However, a vertical ``flip'' of the schedule dramatically changes the average \& tail latency.
    In case~(a) most of the transactions finish early, but in case~(b) most of the transactions wait until the rest of the transactions finish.
    This makes the average and tail latencies~higher.

\begin{figure}
	\centering
	\subfloat[\normalfont{Low Tail Latency}]{
		\includegraphics[width=0.3\linewidth]{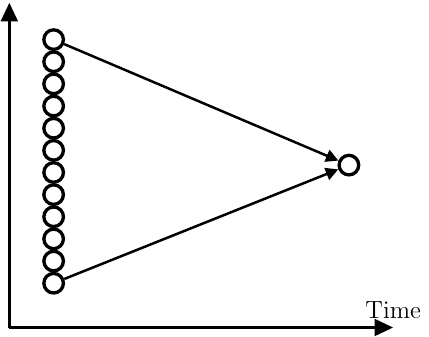}
	}
	\hspace{3em}
	\subfloat[\normalfont{High Tail latency}]{
		\includegraphics[width=0.3\linewidth]{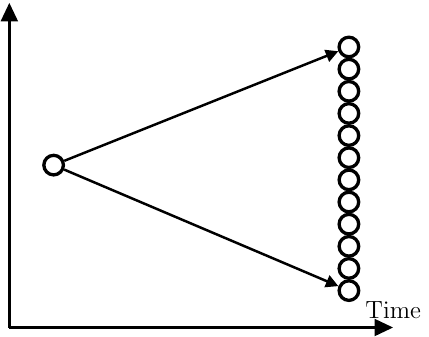}
	}
	\caption{A vertical ``flip'' of \gls{schedule} (a), results in another \gls{schedule} (b) with the same block latency.
	The average \ \& tail latencies\ are vastly different between the two.}
	\label{fig:tail-latency-good-vs-bad}
\end{figure}

    \paragraph{Reordering Colors} \autoref{fig:tail-latency-good-vs-bad} suggests to reorder the colors so that the ``lower'' colors have more transactions than the ``higher'' ones.
    For the homogeneous case, this solution seems viable through \Cref{thm:reorder-coloring-homo} and \Cref{thm:reorder-mincoloring-homo}.
    The former ensures that the block latency is still bounded by the same bound even after reordering the partition (colors), while the latter ensures that the block latency stays optimal even after reordering colors.
    The proofs can be found in~\appref{homotail-proofs}.

    \begin{claim}
        \label{thm:reorder-coloring-homo}
        Given a valid partition of transactions $\T = \sqcup_{i=1}^k T_i$ and a $k$-permutation $\sigma$, let $\Schd$ be the \gls{schedule} created from the \gls{Greedy Schedule} for the given partition $\sqcup_{i=1}^k T_i$ and let $\Schd_\sigma$ be the one for the reordered partition $\sqcup_{i=1}^k T_{\sigma\left(i\right)}$, then
        $ 
            \Lt{\Schd}, \Lt{\Schd_\sigma} \leq k
        $. 
    \end{claim}

    \begin{claim}
        \label{thm:reorder-mincoloring-homo}
        If the partition from \Cref{thm:reorder-coloring-homo} is also a minimal coloring, then
        $ 
            \Lt{\Schd} = \Lt{\Schd_\sigma}
        $. 
    \end{claim}

    \paragraph{Reordering Only Works for Minimal Colorings} \Cref{thm:reorder-mincoloring-homo} raises a question: is the latency preserved after reordering the colors, even if the coloring is not necessarily minimal?
    Although it is a desirable property, block latency can, in fact, change as a result of applying a permutation to the color order.
    Consider the following example with five transactions $1,\ldots, 5$ with their pairwise conflicts described in \autoref{fig:homo-conf-graph}.
    An optimal \gls{schedule} for these transactions with $\Lt{\cdot} = 3$ using the coloring $c_{min}$~(\ref{eq:homo-c_min}) appears in \autoref{fig:homo-cmin-greedy-schedule}.
    \begin{equation}
\small
        \label{eq:homo-c_min}
        c_{\min} = \left\{\begin{array}{@{}c@{}}
                              1\\2
        \end{array}\middle|\begin{array}{@{}c@{}}
                               3\\5
        \end{array}\middle|\begin{array}{@{}c@{}}
                               4
        \end{array}\right\}
\hspace{.4em}
        c_{4} = \left\{\begin{array}{@{}c@{}}
                           1\\2
        \end{array}\middle|\begin{array}{@{}c@{}}
                               3
        \end{array}\middle|\begin{array}{@{}c@{}}
                               4
        \end{array}\middle|\begin{array}{@{}c@{}}
                               5
        \end{array}\right\}
\hspace{.4em}
        \sigma\left(c_{4}\right) = \left\{\begin{array}{@{}c@{}}
                                              1\\2
        \end{array}\middle|\begin{array}{@{}c@{}}
                               3
        \end{array}\middle|\begin{array}{@{}c@{}}
                               5
        \end{array}\middle|\begin{array}{@{}c@{}}
                               4
        \end{array}\right\}
    \end{equation}
        \clonelabel{eq:homo-c_4}{eq:homo-c_min}
        \clonelabel{eq:homo-sigma-c_min}{eq:homo-c_min}
        \clonelabel{eq:homo-sigma-c_all}{eq:homo-c_min}
    \begin{figure*}
        \centering
        \begin{subfigure}{0.2\linewidth}
            \centering
            \includegraphics[width=0.8\linewidth]{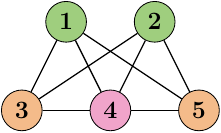}
            \caption{conflict-graph}
            \label{fig:homo-conf-graph}
        \end{subfigure}
        \hfill
        \begin{subfigure}{0.25\linewidth}
            \centering
            \includegraphics[width=0.61\linewidth]{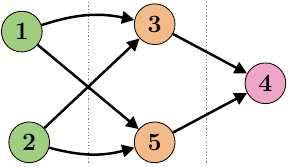}
            \\
            Latency is 3 (Optimal)
            \caption{minimal coloring $c_{min}$}
            \label{fig:homo-cmin-greedy-schedule}
        \end{subfigure}
        \hfill
        \begin{subfigure}{0.25\linewidth}
            \centering
            \includegraphics[width=0.8\linewidth]{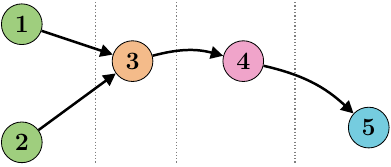}
            \\
            Latency is 4 (Suboptimal)
            \caption{coloring $c_{4}$}
            \label{fig:homo-c4-greedy-schedule}
        \end{subfigure}
        \hfill
        \begin{subfigure}{0.25\linewidth}
            \centering
            \includegraphics[width=0.6\linewidth]{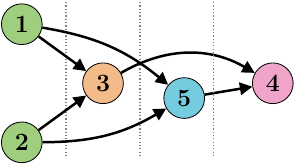}
            \\
            Latency is 3 (Optimal)
            \caption{coloring $\sigma\left(c_{4}\right)$}
            \label{fig:homo-c4-sigma-greedy-schedule}
        \end{subfigure}
        \caption{5 homogeneous transactions and results of the \gls{Greedy Schedule} for them given three colorings~(\ref{eq:homo-sigma-c_all}).}
        \label{fig:homo-greedy-schedules}
    \end{figure*}
    
    Now, consider the non-optimal coloring $c_4$ (\ref{eq:homo-c_4}) and the coloring $\sigma\left(c_{4}\right)$ (\ref{eq:homo-sigma-c_min}) obtained via the 4-permutation $\sigma = \left(3 \quad 4\right)$ that swaps 3 and 4.
    Both $c_4$ and $\sigma\left(c_{4}\right)$ are \emph{not} minimal colorings.
    The resulting \gls{schedule}[s] are depicted in \autoref{fig:homo-greedy-schedules}.
    The \gls{schedule} created from $c_4$ has $\Lt{\cdot} = 4$ as seen in \autoref{fig:homo-c4-greedy-schedule}; but the \gls{schedule} created from $\sigma\left(c_{4}\right)$ has $\Lt{\cdot} = 3$ as seen in \autoref{fig:homo-c4-sigma-greedy-schedule}.
    Interestingly, this schedule is optimal and equivalent to the one in \autoref{fig:homo-cmin-greedy-schedule}.

    \paragraph{Other Corollaries} The example above also shows other interesting facts about homogeneous transactions:
    \begin{enumerate*}[label={\smaller(}\roman*{\smaller)}]
        \item A permutation on the order of the colors may change the block latency.
        \item A non-optimal coloring may yield an optimal \gls{schedule}.
        \item A permutation of its colors can yield a non-optimal \gls{schedule}.
    \end{enumerate*}

    \subsection{Equivalence to Vertex Coloring} \label{sec:color-equiv}

    In \appref{nphardness}, we give a polynomial reduction from the vertex coloring problem to the homogenous latency problem.
    Beyond this reduction, here we also show a direct one-to-one relationship between the colors of a minimal coloring and the structure of an optimal \gls{schedule}; this is  captured by \autoref{thm:minhomolt-eq-chromatic}.
    This strong equivalence relation allows us to directly inherit theoretical results that apply for vertex coloring to the scheduling problem.
    For example, we deduce in this way that for the graph scheduling problem, even $\alpha$-approximation is \gls{NPH}~\cite{color-approx-nph}.

    \begin{theorem}[$\texttt{Lt}^\texttt{*}\equiv_\unitlen\chromaticnumbersymbol$] \label{thm:minhomolt-eq-chromatic}
    For any undirected graph $G =\left(V,E\right)$ ($G$ also represents a homogeneous block when $V$ is the transactions and $E$ is the conflict relation)
    $
    \chromatic{G} = \MinLt[\unitlen]{G}
    $. 
    \end{theorem}

    The proof of this theorem appears in \appref{minhomolt-eq-chromatic}.

    \clearpage

    \section{Heterogeneous Transactions}
    \label{chap:hetero}
    \label{sec:variable-length}

    \subsubsection{Minimal Coloring is not Enough}
    Here, we show why the \gls{Greedy Schedule} procedure does not create an optimal \gls{schedule} given an arbitrary minimal coloring.

    Consider an example with seven transactions, $1,2,3, \ldots, 7$, with durations 1, 10, 100 and 1000 respectively.
    Conflicts and lengths are as illustrated in \autoref{fig:hetero-conf-graph}.
    \autoref{fig:hetero-c1} is the schedule produced from the minimal coloring $c_1$ (\ref{eq:hetero-example-c1}).
    \begin{equation}
        \small
        \label{eq:hetero-example-c1}
        c_1 \! =\! \left\{\!\begin{array}{@{}c@{}}
                         1
        \end{array}\middle|\begin{array}{@{}c@{}}
                               2
        \end{array}\middle|\begin{array}{@{}c@{}}
                               3\\6
        \end{array}\middle|\begin{array}{@{}c@{}}
                               5\\7
        \end{array}\middle|\begin{array}{@{}c@{}}
                               4
        \end{array}\!\right\}
              c_2 = \left\{\begin{array}{@{}c@{}}
                         1
        \end{array}\middle|\begin{array}{@{}c@{}}
                               2
        \end{array}\middle|\begin{array}{@{}c@{}}
                               3\\5
        \end{array}\middle|\begin{array}{@{}c@{}}
                               4\\6
        \end{array}\middle|\begin{array}{@{}c@{}}
                               7
        \end{array}\right\}
        c_3 = \left\{\begin{array}{@{}c@{}}
                                                    1
        \end{array}\middle|\begin{array}{@{}c@{}}
                               2
        \end{array}\middle|\begin{array}{@{}c@{}}
                               3\\6
        \end{array}\middle|\begin{array}{@{}c@{}}
                               4
        \end{array}\middle|\begin{array}{@{}c@{}}
                               5\\7
        \end{array}\right\}
    \end{equation}
        \clonelabel{eq:hetero-example-c2}{eq:hetero-example-c1}
        \clonelabel{eq:hetero-example-c3}{eq:hetero-example-c1}
        \clonelabel{eq:hetero-example-all}{eq:hetero-example-c1}
    This \gls{schedule}'s latency is the maximum weighted path $\Lt{\cdot} = 1+1+1000+100+1000 = 2102$.
    Another minimal coloring $c_2$ (\ref{eq:hetero-example-c2}) can be obtained from $c_1$ by moving 4 and 5 to different colors.
    \autoref{fig:hetero-c2} is the schedule produced from $c_2$, whose latency is the maximum weighted path $\Lt{\cdot} = 1+1+100+1000+100 = 1202$.
    Both schedules have different latencies, even though they were both created from two minimal colorings of the same conflict graph.

    \begin{figure*}
        \centering
        \begin{subfigure}{0.2\linewidth}
            \includegraphics[width=\linewidth]{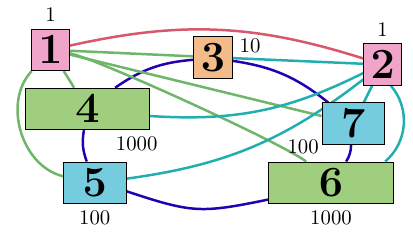}
            \caption{conflict-graph}
            \label{fig:hetero-conf-graph}
        \end{subfigure}
        \hfill
        \begin{subfigure}{0.25\linewidth}
            \centering
            \includegraphics[width=\linewidth]{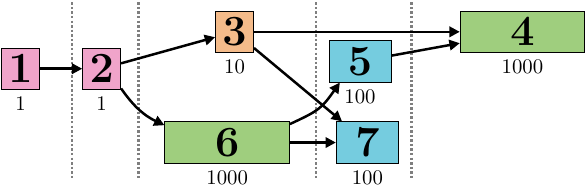}
            \\
            Latency is 2102 (Suboptimal)
            \caption{the minimal coloring $c_{1}$}
            \label{fig:hetero-c1}
        \end{subfigure}
        \hfill
        \begin{subfigure}{0.25\linewidth}
            \centering
            \includegraphics[width=\linewidth]{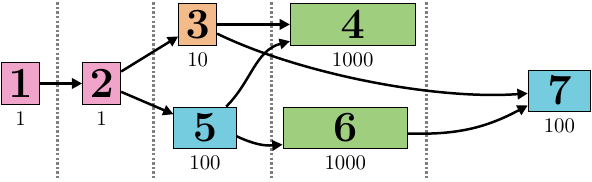}
            \\
            Latency is 1202 (Suboptimal)
            \caption{the minimal coloring $c_{2}$}
            \label{fig:hetero-c2}
        \end{subfigure}
        \hfill
        \begin{subfigure}{0.25\linewidth}
            \centering
            \includegraphics[width=\linewidth]{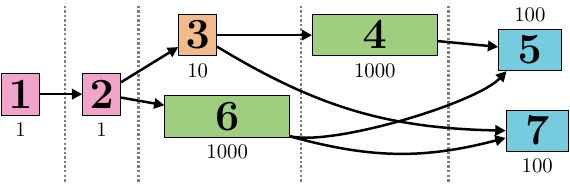}
            \\
            Latency is 1112 (Optimal)
            \caption{the minimal coloring $c_{3}$}
            \label{fig:hetero-c3}
        \end{subfigure}
        \caption{Heterogeneous transactions and the results of the \gls{Greedy Schedule} for them given unweighted colorings (\ref{eq:hetero-example-all}).}
    \end{figure*}

    \paragraph{Color Reordering} In \autoref{sec:lower}, we showed that for \gls{homogeneous transactions}, the optimal latency is preserved even when reordering the colors of a minimal coloring\ (\Cref{thm:reorder-mincoloring-homo}).
    Now, we show that this does not hold for heterogeneous transactions.
    Recall the coloring $c_1$.
    Its colors can be rearranged to create yet another minimal coloring $c_3 = \sigma\left(c_1\right)$ (\ref{eq:hetero-example-c3}) using the 5-permutation $\sigma = \left(4 \quad 5\right)$ that swaps 4 and 5 (and leaves the rest the same).
    \autoref{fig:hetero-c3} is the schedule produced from $c_3$.
    The latency of this second \gls{schedule} is the maximum weight of a path, $\Lt{\cdot} = 1+1+10+1000+100 = 1112$, which is clearly
    optimal for the specific example.
    Thus, rearranging the colors of a minimal coloring can change the latency.

    \paragraph{Tail Latency and Average Latency} In Section~\ref{sec:homo-tail-latency} we described a way to improve the average transaction latency and tail latency by rearranging the colors
    for batches of homogenous transactions.
    However, in the example above, $c_1$ and $c_3$ (\ref{eq:hetero-example-c1})\ represent two isomorphic colorings that only differ in the colors' ``names''.
    Therefore, tactics that try to improve transaction latency by reordering the colors risk turning a schedule with minimal block latency to one with a higher block latency (here, higher than the optimal by~89\%).

    \subsubsection{Weighted Graph Coloring}
    \paragraph{WGCP and Optimal Schedules} \label{example:hetero-weighted-unweighted}We showed above that the algorithm for homogeneous transactions $\MINCOLORA$ does not work for heterogeneous transactions when some arbitrary minimal coloring is used. One idea to overcome this is to turn to the generalized minimal weighted coloring problem~\cite{GUAN97}.
    The goal in the Minimal Weighted Vertex Coloring Problem (aka Minimal WGCP) is to find a coloring that minimizes the sum of the weights of its colors.
    The weight of a color is the maximal weight of a node colored by it.
    Alas, on its own, WGCP cannot be a substitute for GCP (denote the \gls{scheduler} $MWCGBR$).
    As we show here, applying the \gls{Greedy Schedule} to a minimal weighted coloring does not necessarily yield an optimal schedule for heterogeneous transactions.
    Consider four heterogeneous transactions $a$, $b$, $c$, and $d$, whose conflicts and lengths are shown in \autoref{fig:hetero-weightedcoloring-conf-graph}.
    For this graph, it is known that a minimal \emph{weighted} coloring must use 3 colors, although the smallest possible number of colors is 2.
    \autoref{fig:hetero-weightedcoloring-schd-weighted-optimal} shows a schedule s.t. $\Lt{\cdot} = 6$.
    Also, the optimal latency here is $6$, because the conflict $a \conflicts b$ implies a lower bound on the latency, which is the sum of their lengths, $\lenOf{a} + \lenOf{b} = 5+1 = 6$.
    \begin{figure*}
        \centering
        \begin{subfigure}{0.23\linewidth}
            \centering
            \includegraphics[width=\linewidth]{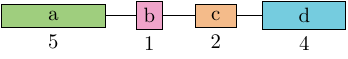}
            \caption{The conflict-graph}
            \label{fig:hetero-weightedcoloring-conf-graph}
        \end{subfigure}
        \hfill
        \begin{subfigure}{0.25\linewidth}
            \centering
            \includegraphics[width=0.75\linewidth]{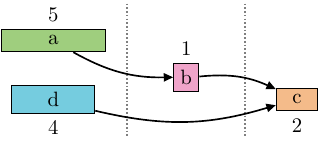}
            \\
            Latency is 8 (Suboptimal)
            \caption{from \emph{weighted}-coloring $c_{w}$}
            \label{fig:hetero-weightedcoloring-schd-weighted}
        \end{subfigure}
        \hfill
        \begin{subfigure}{0.25\linewidth}
            \centering
            \includegraphics[width=0.75\linewidth]{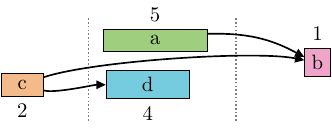}
            \\
            Latency is 6 (Optimal)
            \caption{
                from \emph{weighted}-coloring $c_w'$
            }
            \label{fig:hetero-weightedcoloring-schd-weighted-optimal}
        \end{subfigure}
        \hfill
        \begin{subfigure}{0.25\linewidth}
            \centering
            \includegraphics[width=0.75\linewidth]{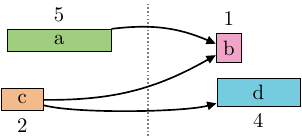}
            \\
            Latency is 6 (Optimal)
            \caption{from coloring $c_{uw}$}
            \label{fig:hetero-weightedcoloring-schd-unweighted}
        \end{subfigure}
        \caption{4 heterogeneous transactions and the results of the \gls{Greedy Schedule} given weighted colorings\ (\ref{eq:minimal-weighted-coloring-all}).}
        \label{fig:hetero-weighted-unweighted}
    \end{figure*}
    It is known that the coloring (\ref{eq:minimal-weighted-coloring}) using 3 colors is a minimal \emph{weighted} coloring for this graph. Denote it by $c_w$.
    \begin{equation}
        \small
                \label{eq:minimal-weighted-coloring}
        c_w = \left\{\begin{array}{@{}c@{}}
                         a\\d
        \end{array}\middle|\begin{array}{@{}c@{}}
                               b
        \end{array}\middle|\begin{array}{@{}c@{}}
                               c
        \end{array}\right\}
                \hspace{1em}
        c_w' = \left\{\begin{array}{@{}c@{}}
                          c
        \end{array}\middle|\begin{array}{@{}c@{}}
                               a\\d
        \end{array}\middle|\begin{array}{@{}c@{}}
                               b
        \end{array}\right\}
                \hspace{1em}
        c_{uw} = \left\{\begin{array}{@{}c@{}}
                            a\\c
        \end{array}\middle|\begin{array}{@{}c@{}}
                               b\\d
        \end{array}\right\}
    \end{equation}
        \clonelabel{eq:reordered-minimal-weighted-coloring}{eq:minimal-weighted-coloring}
        \clonelabel{eq:minimal-unweighted-coloring}{eq:minimal-weighted-coloring}
        \clonelabel{eq:minimal-weighted-coloring-all}{eq:minimal-weighted-coloring}
    The schedule produced from $c_w$ is shown in \autoref{fig:hetero-weightedcoloring-schd-weighted}.
    Its latency is $\Lt[\ell]{\cdot} = 5+1+2 = 8$, and it is not optimal.
    But the reordering $c'_w$ (\ref{eq:reordered-minimal-weighted-coloring}) of the $c_w$ does yield an optimal schedule (\autoref{fig:hetero-weightedcoloring-schd-weighted-optimal}).

    Interestingly, the optimal schedule in \autoref{fig:hetero-weightedcoloring-schd-unweighted} can also be produced from the \emph{unweighted} minimal coloring $c_{uw}$ (\ref{eq:minimal-unweighted-coloring}).
    \clearpage

    \section{Related Work}
    \label{chap:related}

    There is a large body of works on improving the throughput of the transaction ordering mechanism, with many reporting ordering throughput of 10s of thousands and even up to 1 million \gls{tx/s} in certain settings~\cite{hotstuff, FireLedger, Jeolton, Narwhal, ResilientDB, Dumbo, Kauri}.

FastFabric~\cite{FastFabric} uses a combination of optimizations in order to increase the throughput of HyperLedger Fabric (HLF) from $\approx$ 1,000 \gls{tx/s} to $\approx$ 20,000 \gls{tx/s}.
    HLF uses an optimistic concurrency control for executing transactions known as \emph{Execute-Order-Validate} (XOV)~\cite{Alysson2017}.
    Alas, its optimistic concurrency control causes a large percentage of transactions to abort, reducing the effective throughput of the system dramatically.

    Kotla and Dahlin~\cite{KD04} proposed CBASE for parallelizing Byzantine \gls{SMR} by placing a parallelizer between the consensus module and the multithreaded execution layer.
    The
    parallelizer generates a conflicts graph and turns it into a \gls{DAG} according to the transactions' ordering.
    Execution threads execute the transactions according to the \gls{DAG}.
    ParBlockchain~\cite{ParBlockchain} computes a \gls{DAG}, similar to CBASE, but separately per block.

    Enabling concurrent maintenance and access to the \gls{DAG} itself can become a bottleneck too.
    This is addressed in~\cite{COS}, where coarse-grain locking, fine-grain locking, and wait-free implementations are proposed and evaluated.

    Block-STM~\cite{BlockSTM} avoids the need to know transactions' read-sets and write-sets\ by  speculatively executing transactions concurrently and committing a transaction only if it and all of its previous transactions are already validated.
    Invalid transactions are re-executed.

    The above parallelization efforts rely on the total ordering service in the \emph{Order-Execute} (OX) paradigm to break symmetry between conflicting transactions in a deterministic manner, thereby ensuring the replicated state machine semantics.
    Instead, in this work we explored other options that offer greater concurrency potential for the scheduling, such as~coloring.

    Theoretical aspects of scheduling, and in particular variants that try to minimize the total execution time of a batch of tasks with time durations, have been extensively studied.
    Many such works~\cite{AG17,AKLM+15,npc-guide} search for a schedule in which all tasks meat their execution deadline.
    In other cases, conflicts express which jobs cannot be placed on the same machines~\cite{AMPT23,BJW94,ELW16}.
    Other works~\cite{BBB+16,ullman1975np, ullman1973polynomial,Wo00} assume a partial order of dependencies is provided which dictates what tasks must terminate before others start.
    Often, jobs arrive continuously, in which case an online competitive ratio is seeked~\cite{ACL18,AE02,AG17,AKLM+15,AMPT23,APT22,ELW16,ENW02,LST90}.

    In deterministic databases~\cite{DetOverview,Granola,Bohm,PMV,Caracal,Calvin,Aria}, all replicas involved in the execution of transactions must ensure consistent serializability throughout the system.
    The stated motivation for these databases includes lack of concurrency control related to aborts and deadlocks, and simplified commit decisions, replication and failure recovery~\cite{Caracal}.
    In many such systems~\cite{Bohm,PMV,Caracal,Calvin,Aria}, batches include thousands of transactions or more, meaning that the potential maximal concurrency level is higher than the number of available cores.
    For these, defining and proving optimality becomes harder.
    Our lower bounds therefore apply to them too.

    \clearpage
    \section{Discussion} \label{chap:discussions}

    Our fundamental contribution includes identifying an inherent performance vulnerability resulting from carrying the consensus total order onto the serialization order of transactions within a block, and launching a formal study about how to avoid this while ensuring deterministic serializability.
    The change we proposed to the semantics of the total ordering to be projected only onto the blocks, without restricting the logical order of transactions within the same block.
    This new BSXO paradigm, also has substantial implications in the area of speeding up the execution of smart contracts and the overall performance of blockchains.

    We describe a generic framework for \gls{ASMR}, with an injectable \gls{scheduler} responsible for scheduling decisions and execution.
    Our framework is independent of the specific consensus protocol and apply to most of the distributed models ranging from crash failures and up to Byzantine failures.
    It can be adjusted for various block mining and validation, with the added benefit that miners no longer have complete control over the logical order of transactions within a block.
    Graph scheduling allows for representing schedules whose correctness is independent of the architecture used for execution.
    They also allow for defining a universal measurement (latency) that we use for optimizing execution time at the semantic level of the \gls{ASMR}. 
    We prove that the optimization problem is \gls{NPH}.
    We solve the homogenous case using Minimal GCP and prove an equivalence to it.
    We show a technique that helps to reduce average transaction latency and tail latency without compromising the block latency.
    Even though finding a minimal coloring is an \gls{NPH} problem, we believe that it may be possible to use modern SAT solvers~\cite{marijn-solver} for finding minimal colorings.
    Such solvers are quite effective for small to medium instances.
    Further,
in empirical evaluations that we have performed for the greedy coloring algorithm, we have found that for a large selection of graphs, the number of colors it finds is within $1$ of the chromatic number.
    Additionally, its throughput when executing on a standard desktop computer was equivalent to over $500$ Ktps for blocks of up to $5$K transactions.
    This suggests that if we settle for an almost minimal coloring, then the task of coloring is not a performance~bottleneck.
    
    We showed that the heterogeneous case is even more complex to solve and remains largely unsolved.
    Yet, on the positive side, we were able to reduce the problem of optimal scheduling to that of finding level schedules for $\GREEDYA$.

    We also note that in~\appref{batched}, we present a batched execution scheduler, in which all transactions of color $i$ are started immediately after \emph{all} transactions of color $i-1$ finish.
    This serves as another exmple of potential schedulers.
    While this simple model obtains optimal execution time for homogeneous transactions, in the case of heterogeneous transactions, it leads to longer execution times than graph scheduling.
    This model emphasizes the strength of this novel graph scheduling model.

    \clearpage
    \providecommand\makebibliography{\luaexec{require("lib/cls/bbl").currentEngine:print()}}
    \makebibliography
    \clearpage
    \appendix
    \section{The Necessary Number of Phases for Execution}\label{app:bound-proof}
    \paragraph{Formal Motivation.}
    \autoref{fig:gap} shows a specific example for a conflict graph whose maximum concurrency level is as low as two phases, while the maximal number of necessary phases, $n$, can be achieved by requiring a specific logical total ordering.
    Below, we show that situations in which the total ordering order significantly reduces the potential concurrency are quite common.
    
    We use the term \emph{the necessary number of phases for execution} given some total order to refer to the minimum number of phases needed to execute all transactions while still forcing the logical order of the total order.
    A phase cannot contain any conflicting transactions, But transactions can be moved to earlier phases if they do not conflict with any transaction in any of the phases between their original phase and the one they are being shifted to. This also applies to moving transactions to later phases. We can completely eliminate a phase if we move all of its transactions to other phases. 

    When starting the above process on a phase sequence that corresponds to a given total order, we can reduce the number of phases without violating the logical order restriction.
    This process obviously can be executed until, at some point, we can no longer eliminate any of the remaining phases.
    We state that it is always possible to use this process to create a minimum number of phases or, in other words, the necessary number of phases for execution. 
    We use the term \emph{the concurrency level} of the total ordering to refer to $\frac{n}{p}$ when $n$ is the number of transactions, and $p$ is the necessary number of phases for execution while maintaining the logical order of the total ordering.
    The smaller the number of phases, the higher the concurrency level~is.

    In a follow-up work, we study various types of conflict graphs in terms of their longest conflict chain and chromatic number.
    The significance of a long conflict chain is that it enables a malicious miner or validator to launch the following simple performance attack:
    If the order of transactions within the block is consistent with the longest chain $ch$, then adhering to this order implies that the execution would require at least $ch$ phases.
    The larger the ratio between $ch$ and the chromatic number, the more effective this attack is since more potential concurrency is prevented.
    
    Here, we extend the example from  \autoref{fig:gap} to a more broad collection of conflict graphs and also show that there are various orderings that can achieve many levels of concurrency.
    This implies that an honest miner or validator that creates a block and chooses the ordering of the transactions within the block in an arbitrary concurrency oblivious manner, has a non-negligible chance of imposing an order that reduces the level of cocurrency, which is otherwise attainable for these same transactions.
    
    Recall that all transactions are safe to execute concurrently within the same phase.
    Also, if all transactions in a stage take the same time to execute, then the time it takes to execute the entire phase is equivalent to the time it takes to execute a single transaction when enough execution cores are available.
    Hence, in this subsection, we will measure the execution time in terms of the concurrency~level. 

    \begin{theorem}
        Given $n$ transactions and some nontrivial\footnote{A nontrivial conflict relation is one that contains at least one conflict, i.e. not empty.} conflict relation $\conflicts$, let $M$ be the minimum number of phases needed to execute all transactions safely with the maximum concurrency level (with any of the possible total orders)\footnote{Later in the paper, we show that for homogeneous transactions, $M$ is the chromatic number of the corresponding conflict graph.}; 
        and let $ch$ be the maximum length of a conflict chain ($tx_1 \conflicts tx_2 \conflicts tx_3 \conflicts \dots \conflicts tx_{ch}$).

        Then, there are \emph{at least} $\alpha$ distinct concurrency levels for the given collection of transactions where
        $\alpha$ is:
        $$\alpha = \left\lceil \frac{ch - (M-1)}{M-1} \right\rceil.$$
    \end{theorem}
    \begin{proof}
        Assume that the maximal concurrency is deduced from the phases $\varphi_1,\ldots,\varphi_M$.
        The minimality of $M$ implies that all pairs of distinct phases $\varphi_i$ and $\varphi_j$ must have a pair of conflicting transactions.
        Now for any $M \leq k \leq ch$  we attempt to create a total ordering based on a prefix of the given conflict chain: $$tx_1 \to tx_2 \to tx_3 \to \cdots \to tx_k$$ and try to embed the phases $\varphi_1,\ldots,\varphi_M$ into this ordering in a legal way.
        Here, we start with each transaction $tx_1,\ldots,tx_k$ being its own singleton phase.

        Now, we denote the phases $\varphi_1,\ldots,\varphi_M$ after removing the transactions $tx_1,\ldots, tx_k$ as $\varphi'_i \triangleq \varphi_i \setminus \left\{tx_1,\ldots, tx_k\right\}$.
        We look at the phases after removal and distinguish between those that have changed, i.e. $\varphi_i\neq\varphi'_i$, and those that have not, i.e. $\varphi_i=\varphi'_i$. 
        We call the former group of phases \emph{mergeable phases} and the latter \emph{residual phase}s; and denote the residual phases as $\varphi_{r_1},\ldots,\varphi_{r_t}$ and the remaining mergeable phases as $\varphi'_{m_1},\ldots,\varphi'_{m_{M-t}}$ when $t$ is the number of residual phases.

        We now focus on producing a phase sequence from which we will derive a total order.
        The first step we take is to merge each of the mergeable phases $\varphi'_{m_i}$ with \emph{one} of the singletone phases represented by the transaction $tx_i$ such that $tx_i$ was originally in it (i.e., $tx_i \in \varphi_{m_i}$).
        Each one of the merge operations is safe since it does not create phases that contain conflicts.
        This is because the merge $\varphi'_{m_i}\cup\left\{tx_j\right\}$ is a subset of $\varphi_{m_i}$ that is conflict-free in itself.
        The second step we take is to add the residual phases at the beginning of the phase sequence.
        This creates the phase sequence $P=P_1,\ldots,P_t,P_{t+1},\ldots,P_{t+k}$ that \autoref{fig:motivation-proof-embeddings} depicts.
        \ifcoloredproof 
        \begin{center}
        \else 
        \begin{figure}[h]
        \centering
        \fi 
            \centering
            \includegraphics[width=0.8\linewidth]{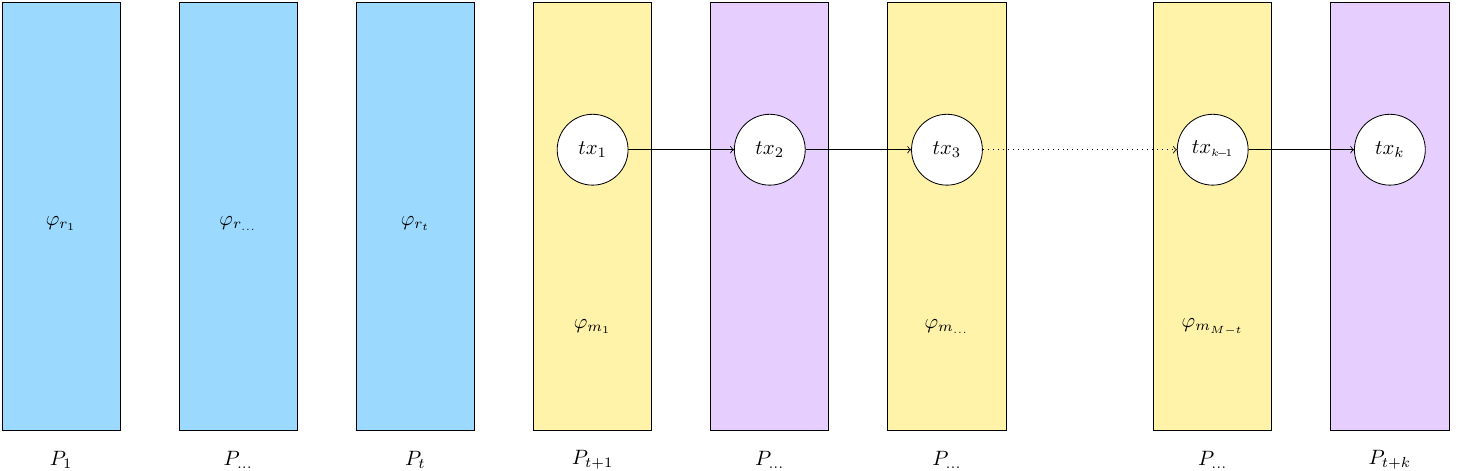}
            \ifcoloredproof 
                \captionof{figure}{Illustration of creating the phase sequnece with $t+k$ phases.}
            \else 
                \caption{Illustration of creating the phase sequnece with $t+k$ phases.}
            \fi 
            \label{fig:motivation-proof-embeddings}
        \ifcoloredproof 
        \end{center}
        \else 
        \end{figure}
        \fi 

    The phase sequence $P$ may be reduced to phase sequences that use fewer phases using the phase elminimation techniques we described in the above paragarphs. 
    After the elimination of phases, we reach some \emph{minimal} phase sequence $P^*=P^*_1,\ldots,P^*_\ell$.
    We use \autoref{fig:motivation-proof-smashed} to illustrate the minimal phase sequence $P^*$ that is created after eliminating some of the phases in $P$.
    We use $\ell$ to denote the number of phases in $P^*$.
    Obviously, $\ell \geq k$ since the conflict chain $tx_1 \conflicts tx_2 \conflicts tx_3 \conflicts \dots \conflicts tx_{k}$ cannot be executed in less than $k$ phases.
    Furthermore, the residual phases cannot be reduced to fewer phases due to the minimality of $M$.
    
    Thus, the only way for $\ell$ to be strictly less than $k+t$ is for us to merge the suffix of the residual phases $P_1,\ldots,P_t$ with a prefix of the other phases $P_{t+1},\ldots,P_{t+k}$; some of the transactions can be moved between phases; however, the conflict chain must be preserved with $k$ distinct phases.
    In total, the minimal phase sequence $P^*$ has a length of $\ell$ phases when $\ell=k+t_k$ for some $t_k\in\left\{0,\ldots,t\right\}$.
    For each $M \leq k \leq ch$, we use $t_k$ to denote the value determined by the length of $P^*$ created in the process described above. 

        \ifcoloredproof 
        \begin{center}
        \else 
        \begin{figure}[h]
        \centering
        \fi 
            \includegraphics[width={\dimexpr 0.68\linewidth}]{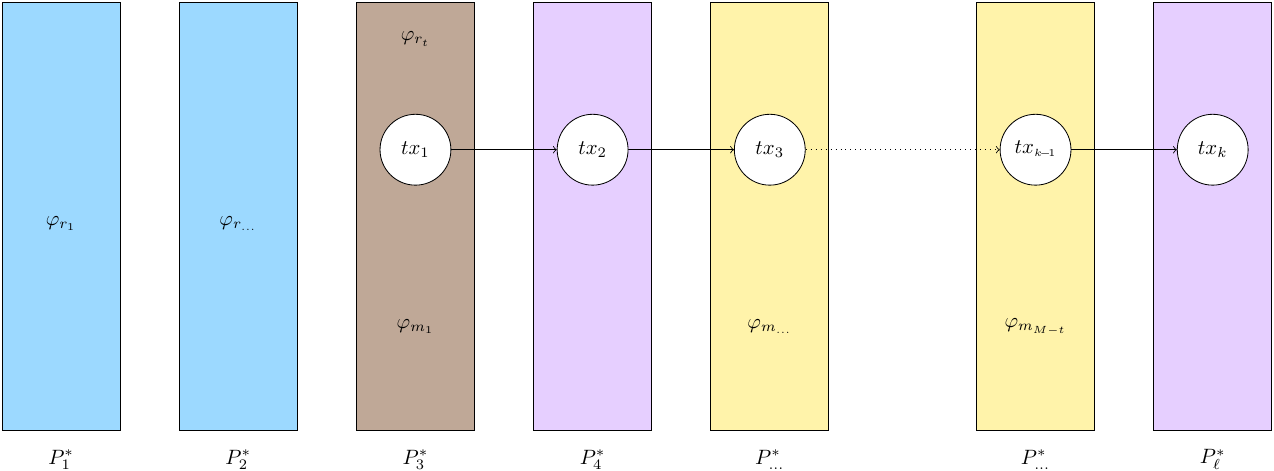}
            \ifcoloredproof 
                \captionof{figure}{Illustration of the compated phases sequence.}
            \else 
                \caption{Illustration of the compated phases sequence.}
            \fi 
            \label{fig:motivation-proof-smashed}
        \ifcoloredproof 
        \end{center}
        \else 
        \end{figure}
        \fi 

    Now, we attempt to estimate combinatorially how many distinct values $\left(k+t_k\right)$ there are.
    Note that for any $M \leq k \leq ch$, the value of $t_k$ must be at most the number of residual phases.
    We also know that this number cannot exceed $M-2$: In a non-trival conflict relation, the longest conflict chain must include at least one conflict between two transactions; each of these two transactions must be included in a different phase $\varphi_i$; therefore, there must be at least two mergeable phases.
    In other words, we are trying to estimate the number of values in the set $\left\{k+t_k\mid k=M,\ldots,ch\right\}$ when $t_k\in\left\{0,\ldots,M-2\right\}$.
    
        \ifcoloredproof 
        \begin{center}
        \else 
        \begin{figure}[h]
        \centering
        \fi 
            \includegraphics[width=\linewidth]{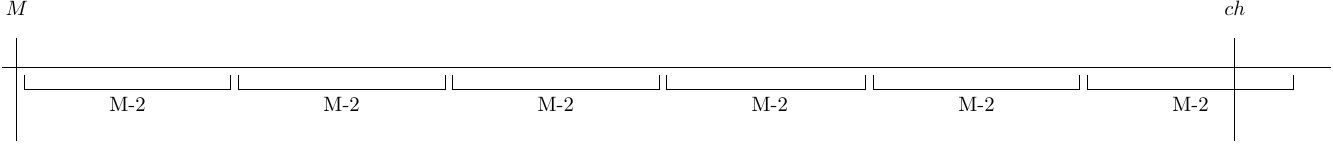}
            \ifcoloredproof 
                \captionof{figure}{Estimating the minimal number of distinct values for $\left\{ k+t_k \right\}$}
            \else 
                \caption{Estimating the minimal number of distinct values for $\left\{ k+t_k \right\}$}
            \fi 
            \label{fig:motivation-minimize}
        \ifcoloredproof 
        \end{center}
        \else 
        \end{figure}
        \fi 

        \autoref{fig:motivation-minimize} illustrates how to reach the minimum number of distinct values in the given scenario.
        This is done by collapsing every $M-1$ consecutive values of $k$ into one distinct value of $k+t_k$ by assigning $t_k$ the respective values $M-2, M-1, \ldots, 1, 0$.
        Thus, the number of distinct values is at least $\alpha$ when $$\alpha = \left\lceil \frac{\text{\# of values between }M\text{ and }ch}{\text{\# of values between }0\text{ and }M-2} \right\rceil = \left\lceil \frac{ch - M + 1}{(M-2) + 1} \right\rceil = \left\lceil \frac{ch - (M-1)}{M-1} \right\rceil$$
        
        Of course $\alpha$ is also a lower bound on the number of concurrency levels, by definition.
    \end{proof}

\section{On the Performance Vulnerability's Potential Damage}
\label{app:attack}

\subsection{Finding Longest Paths}
\srdsref{fig:gap} and \srdsref{sec:vulnerability} suggest that requiring a specific total order can result in inherently poor parallelism, especially when specially crafted by a malicious block creator.
In \appref{bound-proof}, we prove that the existence of a simple path of $k$ vertices in the conflict graph corresponds to total orders that imply a minimal latency of at least $k$ (for homogeneous blocks).
We also know from \srdsref{chap:homo} that the chromatic number of the conflict graph is the lowest possible latency for the corresponding set of transactions and conflicts.

The combination of these two facts can be used to evaluate the ratio between the worst latency and the best latency for a given block.
That is, assuming $\ell$ is a lower bound for the longest simple path in a given conflict graph and $\chi$ is an upper bound for its chromatic number, we can use $\frac{\ell+1}{\chi}$ as a lower bound on the ratio between the worst possible latency and the best attainable latency.

\subsection{Empirical Estimation}
\autoref{app:alg:simulation} uses a combination of a greedy coloring algorithm and a ``random-walk'' algorithm for heuristically searching for a longest simple path (the \textproc{est-$\ell$ path} function), based on the ideas presented above.
This corresponds to an honest, but uninformed, miner choosing an arbitrary total order for the block.
Hence, this algorithm is also an empirical estimation for the expected ratio as well as a \emph{non-strict} lower bound.

\subsection{Empirical Evaluation}
\begin{figure}[h]
    \begin{center}
        \hspace{-2.6em}
    \includegraphics[width=1.08\linewidth]{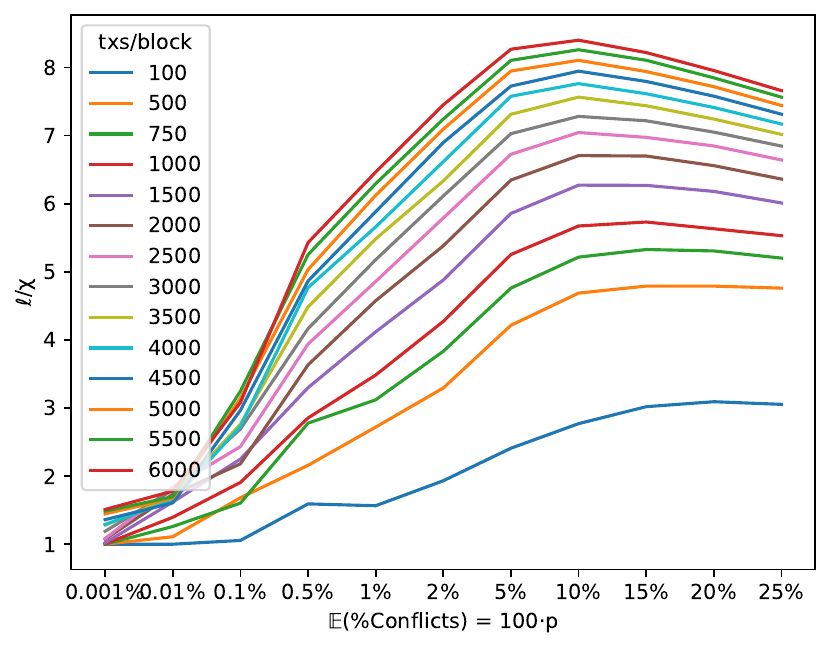}
    \end{center}
    \caption{Empirical lower-bound estimation of the ratio between the length of a longest simple path and the chromatic number.}
    \label{fig:gnp_k100}
\end{figure}

We executed \autoref{app:alg:simulation} using $100$ samples for each data point. 
We used 14 different block sizes ranging from $100$ txs/block and up to $6,000$ txs/block.
We randomly selected $G(n,p)$ graphs for the sample conflict graphs used by the algorithm, while using various values of $p$.
Since it is known that for $G(n,p)$ graphs, $n\cdot p$ is the expected number of edges in the graph, we chose the values of $p$ to represent the percentage of conflicts in the graph.
We tested values of $p$ that correspond to conflict percentages ranging from $0.001\%$ up to $25\%$.
The exact values we used are detailed in the X-axis of \autoref{fig:gnp_k100}.

\begin{algorithm}
        \begin{algorithmic}[1]
        \Function{est-$\ell$path}{$G=\left( V,E \right)$}
            \State \textbf{order} vertexes in $V$ as $v_1,v_2,\ldots,v_{n=\left|V\right|}$
            \State $\ell \gets (0,\overset{n}{\ldots\ldots},0)$
            \State $P \gets (v_1,\overset{n}{\ldots\ldots},v_n)$
            \For{$i=1,\ldots,n$}
                \For{$j=i+1,\ldots,n$}
                    \If{$(v_i,v_j) \in E$}
                        \State $\ell[j] \gets \max\left( \ell[j], \ell[i] + 1 \right)$
                        \State \algorithmicif\ $\ell[j] < \ell[i] + 1$ \algorithmicthen\ $P[j] \gets (P[i] \parallel v_j)$
                    \EndIf
                \EndFor
            \EndFor
            \State \Return $\max(\ell)$
        \EndFunctionExplicit
        \Function{est-$\chi$}{$G=\left( V,E \right)$}
            \State \Return \Call{GreedyColoring}{$G$} \textbf{using} vertex  order\Statex\hspace{4em} by descending degree
        \EndFunctionExplicit
        \Statex
        \State $R \gets []$
        \Repeat
            \State $G \gets$ random conflict graph
            \State $r \gets \left(\dfrac{\Call{est-$\ell$ path}{G}+1}{\Call{est-$\chromaticnumbersymbol$}{G}}\right)$
            \State $R \gets R \parallel [r]$
        \Until reached $K$ samples
        \State \textbf{yield} $\text{avg}(R)$
    \end{algorithmic}
    \caption{Pseudocode for empirical simulation}
    \label{app:alg:simulation}
\end{algorithm}

\subsection{Conclusions}
\autoref{fig:gnp_k100} depicts the results of this experiment.
The experiment shows that as the block size grows so does the ratio, for all values of $p$.
Also, the ratio is larger than $1$ for all instances expect for a few cases when both $p$ and the number of transactions are both small (the bottom left corner of \autoref{fig:gnp_k100}).
When the percent of conflicts is at least $0.5\%$, the difference is at least $56.50\%$ (for $p=1\%$ and $100$ tx/block). 
The difference peaks at $\bm{7,400.84\%}$ (for $p=10\%$ and $6,000$ txs/block).
On average the difference was $\bm{371.02\%}$.
Hence, even a non-malicious block creator that places transactions in a random order inside its blocks may cause non-negligible performance loss if the logical serialization order must respect this unintentional order.

\section{Generic Framework Definition \& Proofs}
\label{app:asmr-framework}
In our framework, blocks' transactions are executed using a swappable \gls{scheduler},
as defined~below:
    \begin{definition}[\Glsxtrshort{scheduler}]
        A \Gls{scheduler} $BR$ is an algorithm used to execute transactions in a block.
        It provides its functionality using three components:
        \begin{enumerate}
            \item \underline{Scheduling Decisions}: Defines how to create the schedule deterministically based on the block.
            It exports the following two methods:
            \begin{itemize}
                \item $BR\!::\!\textproc{make-schedule}(\texttt{transactions}, \texttt{constra}\-\texttt{ints}) \to \texttt{schedule}$ computes the schedule given the constraints (e.g. conflicts)
                \item $BR\!::\!\textproc{validate-schedule}(\texttt{transactions}, \texttt{cons}\-\texttt{traints}, \texttt{schedule}) \to \texttt{boolean}$ validates a schedule given the constraints
            \end{itemize}

            \item \underline{Scheduler}: Performs the execution of the transactions according to the given schedule.
            It handles operations invoked on objects in the global state; it also ensures that the invocation of an operation on an object considers the most recent version of the object according to the logical order of operations.
            It uses the following two methods:
            \begin{itemize}
                \item $BR\!::\!\textproc{init-execution}(\texttt{schedule}, \texttt{global sta}\-\texttt{te}) \to \texttt{execution}$ initializes the required data-structures for the execution of the transactions according to the schedule.
                \item $BR\!::\!\textproc{start-execution}(\texttt{execution})$ begins the execution.
            \end{itemize}

            \item \underline{Runtime}: Returns the results of the execution of transactions and the changes to the global state.
            Uses the following three methods:
            \begin{itemize}
                \item $BR\!::\!\textproc{is-execution-running}(\texttt{execution}) \to \texttt{boolean}$ test if the execution is still running or if transaction results have not been received yet
                \item $BR\!::\!\textproc{next-execution-results}(\texttt{execution}) \to \texttt{results}$ emits transactions' results as soon as the transactions finish executing
                \item $BR\!::\!\textproc{state-changes}(\texttt{execution}) \to \texttt{state c}\-\texttt{hanges}$ emits the global state changes
            \end{itemize}
        \end{enumerate}
    \end{definition}
    \begin{definition}[Execution of a Block]
        Given a \gls{scheduler} $\A$, an execution of $\A$ for some block $\Block$ is any of the executions of \textsc{ExecuteBlock} in \autoref{alg:exec-block} 
        given the block $\Block$ and some state.
    \end{definition}
    \begin{algorithm}[]
        \caption{\small Pseudocode for the main loop given the \gls{scheduler} $\A$}
        \label{alg:main-loop}
        \footnotesize
\begin{algorithmic}[1]
    \State $s \gets$ current global state
    \State Init consensus layer $CL$
    \Loop
        \State $\Block \gets$ next unprocessed block from $CL$ \label{alg:main-loop:fetch}
        \If{$\Block$ is not valid}
            \State \textbf{emit} errors for all transactions in $\Block$
            \State \textbf{go to} line~\ref*{alg:main-loop:term-batch}
        \EndIf
        \State state-changes, results $\gets \Call{execute-block}{\Block, s}$ using $\A$
        \State \textbf{send} transaction results to clients as necessary
        \State $s \gets s$ after \textbf{applying} state-changes
        \State \textbf{save} $s$ as the current global state
        \State \textbf{mark} block $\Block$ as processed \label{alg:main-loop:term-batch}
    \EndLoopExplicit
    \algstore{alg1}
\end{algorithmic}
        \label{alg:exec-block}
\newcommand{\C}{\mathcal C}
\begin{algorithmic}[1]
    \algrestore{alg1}
        \Procedure{execute-block}{$\Block$, state}\label{alg:main-loop:exec-block:start}
            \State $\T \gets $ transactions of block $\Block$
            \State $\C \gets $ prepare constraints for $\T$ of block $\Block$
            \State $\Schd \gets BR\!::\!\Call{make-schedule}{\T, \C}$
            \State exec $\gets BR\!::\!\Call{init-execution}{\Schd, \text{state}}$
            \State $BR\!::\!\Call{start-execution}{\text{exec}}$
            \While{$BR\!::\!\Call{is-execution-running}{\text{exec}}$} \label{alg:main-loop:exec-block:while-results}
                \State\textbf{emit} $BR\!::\!\Call{next-execution-results}{\text{exec}}$ as transactions results \label{alg:main-loop:exec-block:client-results}
            \EndWhile
            \State state-changes $\gets BR\!::\!\Call{state-changes}{\text{exec}}$
            \State \Return state-changes, all transaction results
        \EndProcedureExplicit \label{alg:main-loop:exec-block:end}
\end{algorithmic}
    \end{algorithm}
    \subsubsection{Proof of Correctness}
    \begin{lemma}
        \label{thm:detr-critirion2}
        \gls{scheduler} $\A$ is \emph{sequentially deterministic}, if and only if,
        for every block~$\Block$:
        \begin{enumerate*}[label={\smaller(}\roman*{\smaller)}]
            \item Every execution of $\A$ for $\Block$ is \gls{serializable} and,
            \item \Gls{scheduler} $\A$ is deterministic.
        \end{enumerate*}
    \end{lemma}

    \begin{observation}
        \label{thm:detr-serial-equiv}
        Given a \gls{scheduler} $\A$ that is sequentially deterministic and some block $\Block$.
        Then, there is some serial execution $\tau_\Block$ for block $\Block$ such that for any execution $\sigma$ of $\A$ for $\Block$, $\sigma \equiv \tau_\Block$.
    \end{observation}
    \begin{theorem}
        \label{thm:main-determinisitic}
        Assume that \gls{scheduler} $\A$ is deterministic.
        Then, \autoref{alg:main-loop} is also deterministic, i.e., for every possible (infinite) sequence of blocks, all executions generated by \autoref{alg:main-loop} for the given block sequence are pairwise equivalent.
    \end{theorem}
    \begin{proof}
        Let $\sigma$ and $\sigma'$ be two executions of \autoref{alg:main-loop} using $\A$ for a given sequence of blocks $\left(\Block_i\right)_{i=1}^\infty$.
        Observe that for any pair of transactions $tx_1$ and $tx_2$, if $\earlierthan{tx_1}{tx_2}$, then the block in which $tx_1$ appears must be before the block of $tx_2$ in the total ordering (lines \ref{alg:main-loop:term-batch} and \ref{alg:main-loop:fetch}).
        In addition, all transactions belonging to the same block $\Block$ are concurrent, by definition, and line~\ref{alg:main-loop:fetch}.
        Therefore, each block $\Block$ is a contiguous set of transactions that are concurrent in $\sigma$.
        We can deduce from here that there exist two corresponding equivalent
        executions $\tau = \tau_{\Block_1}\cdot\tau_{\Block_2}\cdot\ldots\equiv\sigma$ and $\tau' = \tau'_{\Block_1}\cdot\tau'_{\Block_2}\cdot\ldots\equiv\sigma'$ such that for each block $\Block_i$, both $\tau_{\Block_i}, \tau'_{\Block_i}$ are two
        executions equivalent to executions of $\A$ for block $\Block_i$.

        Since $\A$ is deterministic, we know that $\forall_{\Block_i} \: : \: \tau_{\Block_i} \equiv \tau'_{\Block_i}$.
        Thus, from the equivalence of concatenation and the assumptions, we know that
        $ 
            \sigma \equiv
            \tau =
            \tau_{\Block_1}\cdot\tau_{\Block_2}\cdot\ldots \equiv
            \tau'_{\Block_1}\cdot\tau'_{\Block_2}\cdot\ldots =
            \tau' \equiv
            \sigma'
        $ 
        and from transitivity, we can deduce that $\sigma \equiv \sigma'$ as necessary.
    \end{proof}

    \autoref{thm:main-determinisitic} implies that when using a deterministic \gls{scheduler}, all replicas compute the same results for each transaction, and thus the framework makes for a valid replicated service.

    \begin{theorem}
        \label{thm:main-serializable}
        Assume that \gls{scheduler} $\A$ is serializable.
        Then, every execution generated by \autoref{alg:main-loop} is strictly \gls{serializable} when using a single replica.
    \end{theorem}
    \begin{proof}
        To prove the theorem, we need to show that for every execution $\sigma$ generated by \autoref{alg:main-loop} there exists a sequential execution $\tau$ that obeys $\sigma$.
        We claim that extending the order in which transactions are executed in $\sigma$ so that all transactions in the same block are placed one after the other is the desired sequential execution $\tau$.

        Using similar arguments to the ones made in the beginning of the proof~of~\autoref{thm:main-determinisitic} we know that the subexecution of $\sigma$ obtained only by considering transactions of some block $\Block$, denoted $\sigma_\Block$, is also a possible execution of~$\A$, for the block $\Block$.

        Now, consider some block $\Block$.
        By \autoref{thm:detr-critirion2}~(i) for $\A$, there exists an equivalent serial execution $\tau_\Block$ that is equivalent to the subexecution $\sigma_\Block$.
        Thus, we can construct $\tau$ incrementally by appending $\tau_\Block$ to $\tau$ for each block $\Block$ in the order produced by the consensus layer.

        The order of the transactions in $\tau$ respects the partial order of $\sigma$ because the construction of $\tau$ respects the block order.
        It is also easy to see that $\tau \equiv \sigma$ because the concatenation of equivalent executions produces equivalent executions.
        In summary, $\tau$ is a valid sequential execution that is equivalent to $\sigma$, and therefore $\sigma$ is strictly \gls{serializable}~\cite{DB-textbook}.
    \end{proof} 

    \subsubsection{Proof of Completeness}
    \begin{theorem}
        \label{thm:main-determinisitic-anti}
        Assume that \autoref{alg:main-loop} is deterministic, that is, for every possible (infinite) sequence of blocks, all executions generated by \autoref{alg:main-loop} for the given block sequence are pairwise equivalent.
        Then, \gls{scheduler} $\A$ is deterministic.
    \end{theorem}
    \begin{proof}
        We prove this by showing that for any given state $s$ and any given block $\Block$, all executions of $\A$ for $\Block$ are equivalent.
        Fix some $s$ and some $\Block$.
        Let two executions of $\A$, $\sigma$, and $\sigma'$ for block $\Block$.

        It is obvious that $\sigma$ and $\sigma'$ can be converted into equivalent corresponding executions of \autoref{alg:main-loop}, $\tau$ and $\tau'$, given the block sequence containing only the single block $\Block$.
        By the assumption we know that $\tau\equiv\tau'$, and thus $\sigma\equiv\sigma'$.
    \end{proof}

    \begin{theorem}
        \label{thm:main-serializable-anti}
        Assume that every execution generated by \autoref{alg:main-loop} is strictly \gls{serializable} when using a single replica.
        Then, \gls{scheduler} $\A$ is serializable.
    \end{theorem}
    \begin{proof}
        We prove this by showing that for any given state $s$ and any given block $\Block$, all executions of $\A$ for $\Block$ are serializable.
        Fix some $s$ and some $\Block$.
        Let some execution $\sigma$ of $\A$ for $\Block$.

        Here again, it is obvious that $\sigma$ can be converted into an equivalent corresponding execution of \autoref{alg:main-loop}, $\tau$, given the block sequence that contains only one block, $\Block$.
        By the assumption, we know that $\tau$ is serializable, and thus $\sigma$ is also serializable.
    \end{proof} 

    Using these two theorems, we know that \emph{deterministic sequentiality} is the weakest condition required for guaranteeing strict serializability of the framework.

\section{Missing Proofs from \srdsref{sec:greedy}}
\label{app:greedy-proofs}

\subsection{Proof of \srdsref{lemma:valid-greedy-schedule}}
\label{app:valid-greedy-proof}

    \begin{proof}
        Let two transactions be such that $tx_i \conflicts tx_j$. \Gls{wlog} assume that $tx_i \in B_i \neq B_j \ni tx_j$.
        Since $tx_i \conflicts tx_j$, some iteration of \hyperref[alg:greedy-schedule:iter_start]{lines~\ref*{alg:greedy-schedule:iter_start}}~--~\ref{alg:greedy-schedule:iter_end} checks if $\Schd$ already contains a path between them and otherwise adds a direct edge between them (line~\ref{alg:greedy-schedule:iter_end}).
        Moreover, for two $tx, tx' \in B_i$ we know that $tx \notconflicts tx'$ because $B_i$ is \gls{conflict-free}.
        Thus, for any two transactions $tx_i \conflicts tx_j$ the \gls{schedule} contains a directed path between them, so the \gls{schedule} is valid.
    \end{proof}

\subsection{Proof of \srdsref{thm:convert-to-partition}}
\label{app:convert-to-partition}
    \begin{proof}
        For the proof, we assign each vertex $v\in V$, the (highest) color that represents the depth of its corresponding node in $\GSchd$ using the procedure \textproc{ConvertToColoring}.
        The construction of $c$ is as follows: First, all nodes $v\in V$ with $\indeg[\GSchd]{v} = 0$ get the color $c\left(v\right) := 1$.
        Next, all nodes $v\in V$ with an incoming edge $\left(u, v\right)$ s.t. $c\left(u\right) = 1$ get the color $c\left(v\right) := 2$.
        We repeat this process until all the vertices are given a color.
        Note that in the process we described, some vertices may theoretically be assigned multiple colors.
        To make the process well-defined, we give each vertex the \emph{highest} color it can get from the process above, i.e. its \emph{depth}.
        Recall that in a \gls{DAG} there is a finite number of possible paths and that $\GSchd$ is a \gls{DAG}; thus, the coloring is well-defined.

        Now, the coloring $c$ must be legal because for any two given vertexes $v, v' \in V$ with an undirected edge $\left(v, v'\right) \in E$ we also know that $v \conflicts v'$.
        By the validity of $\Schd$, there must be a directed path between them in $\confGraph{T}$; \gls{wlog} there is a path $v \longrightsquiglearrow_{\GSchd} v'$ (and $v' \not\mathrel{\longrightsquiglearrow}_{\GSchd} v$).
        Thus, for any path $ * \longrightsquiglearrow_{\GSchd} v$, we can create a strictly longer path $ * \longrightsquiglearrow_{\GSchd} v \longrightsquiglearrow_{\GSchd} v'$.
        Therefore, $c\left(v\right) < c\left(v'\right)$.
    \end{proof}

\subsection{Proof of \srdsref{thm:greedy-Completeness}}
\label{app:greedy-main-proof}

    \begin{proof}
        Using \Cref{thm:convert-to-partition}, we create a legal partition $\bigsqcup_{i=1}^k B_i = \T$ from $\Schd$ using the procedure \textproc{ConvertToColoring}.
        Each set $B_i$ represents the $i$-th color, i.e., all transactions that were assigned with the number $i$.
        \newcommand{\GS}{\Schd'}
        Now denote the level schedule produced from the partition above $\GS = \textproc{LevelSchedule}\!\left(\sqcup_{i=1}^k B_i\right)$.

        We show that (\ref{eq:transative-path}) holds, i.e., every path $P=\left(v_1, v_2, \ldots, v_l\right) \in \GS$ is contained in $\GSchd$ as a transitive sub-path ($\forall i\;\; v_i \rightsquiglearrow_{} v_{i+1}$).
        \newcommand{\squigarr}{\longrightsquiglearrow}
        \begin{equation} \label{eq:transative-path}
        \small 
        \begin{gathered}
            v_1 \rightarrow_{\GS} v_2 \rightarrow v_3 \rightarrow_{\GS}
            \cdots
            \rightarrow_{\GS} v_{l-2} \rightarrow_{\GS} v_{l-1} \rightarrow_{\GS} v_l
            \\
            \Downarrow
            \\
            v_1 \squigarr_{\Schd} v_2 \squigarr v_3 \squigarr_{\Schd}
            \cdots \squigarr_{\Schd} v_{l-2}
            \squigarr_{\Schd} v_{l-1} \squigarr_{\Schd} v_l
        \end{gathered}
        \end{equation}
        Fix some path $P=\left(v_1, v_2, \ldots, v_{l-1}, v_l\right) \subseteq \GS$ and consider an arbitrary edge $v_i\rightarrow_{\GS} v_{i+1}$ in $P$.
        Now we prove that $v_i\longrightsquiglearrow_{\Schd} v_{i+1}$.

        Since $\GS$ was created using the \gls{Greedy Schedule}, we know that the edge $v_i\rightarrow v_{i+1}$ could only have been added to $\GS$ as a result of line~\ref{alg:greedy-schedule:confpairs}, for some two sets $B_m\ni v_i$ and $B_M\ni v_{i+1}$ s.t.~$m < M$.
        Thus, $v_i \conflicts v_{i+1}$ and because $\Schd$ is a valid schedule inducing a \gls{DAG}, we know that either $v_i\longrightsquiglearrow_{\Schd} v_{i+1}$ or $v_{i+1} \longrightsquiglearrow_{\Schd} v_i$.
        Suppose, for the sake of contradiction, that $v_{i+1} \longrightsquiglearrow_{\Schd} v_i$.
        According to the procedure \textproc{ConvertToColoring}, $v_i \in B_m$ implies that $\depth[\GSchd]{v_i} = m$\footnote{The depth is the length of the longest simple path (that ends in $v_i$) in vertices.}.
        Similarly, $\depth[\GSchd]{v_{i+1}} = M$.
        Therefore, $\depth[\GSchd]{v_i} = m < M = \depth[\GSchd]{v_{i+1}}$.
        On the other hand, $\GSchd$ is a \gls{DAG} so if $v_{i+1} \longrightsquiglearrow_{\Schd} v_i$ then $\depth[\GSchd]{v_i} > \depth[\GSchd]{v_{i+1}}$, and we reach a contradiction \lightning.

        In summary, for every path in $\GS$, (\ref{eq:transative-path}) holds. Now we show that the latency of $\Schd'$ is no worse than the latency of $\Schd$.
        \newcommand{\PS}{\mathcal{P}}
        \newcommand{\PGS}{\mathcal{P}'}
        For convenience, denote the sets of all simple paths in $\GSchd$ and $\SchG{\GS}$ as the corresponding $\PS$ and $\PGS$.
        Since (\ref{eq:transative-path}) holds, we know that:
        \begin{equation} \label{eq:transative-path-subsets}
        \forall P_{\GS} \in \PGS \ \exists P_{\Schd} \in \PS \: : \: P_{\GS} \subseteq P_{\Schd}.
        \end{equation}

        Therefore, in total we can use (\ref{eq:transative-path-subsets}) to split a maximal path in $\PS$ into two a sum of two (positive-summed) disjoint sets; one containing all vertices from a maximal path in $\PGS$ and the other containing the rest.
        We can then deduce the corresponding inequality. 
        Thus, $\Lt[\len]{\Schd} \geq \Lt[\len]{\GS} = \Call{LevelSchedule}{\sqcup_{i=1}^k B_i}$ as required.
    \end{proof}

\section{NP-Hardness (Missing Proofs from \srdsref{sec:nphardness})}
\label{app:nphardness}
    In this part, we prove that determining minimal latency and optimizing it is \gls{NPH}.
    This is done by reduction from the $\symb{Color}$~(\ref{def:color}) Graph Vertex Coloring problem.
    \begin{equation}
        \label{def:color}
        \symb{Color} \triangleq \big\{\left(G, k\right) \Big|\, \exists \: c \text{ coloring of } G \text{ using less than } k \text{ colors} \big\}
    \end{equation}

    \paragraph{Transformation Function.} The function \textproc{transform}, given an undirected graph $G$, creates an input for the scheduling problem.
    \begin{algorithmic}[1]
    \footnotesize 
        \State \textbf{given} some arbitrary constant, $c\in\symb{N+}$
        \Function{transform}{$G = \left(V, E\right)$ : Undirected Graph}
            \State \textbf{let} $\left(\T,\: \conflicts,\: \len:\T\to\symb{N+}\right)$
            \ForAll{$v \in V$}
                \State $\T \gets \T \cup \left\{v\right\}$
                \State $\lenOf{v} \gets c$ \Comment{Assign the same length for all transactions}
            \EndFor
            \ForAll{$\left(u, v\right) \in E$}
                \State \textbf{fix} $u \conflicts v$\Comment{Adjust the relation $\conflicts$ accordingly}
            \EndFor
            \State \Return $\left(\T, \conflicts, \len\right)$
        \EndFunctionExplicit
    \end{algorithmic}

    \begin{lemma}\label{lemma:transform-poly}
        The function \textproc{transform} is computable in polynomial time in $\left|G\right|$ (the size of the encoding of $G$).
    \end{lemma}
    \begin{proofsketch}
        The function makes a linear pass over all vertices and edges, and for each item the iteration takes $\Theta\left(1\right)$ time. $\square$
    \end{proofsketch}

    \begin{lemma} \label{lemma:greedy-schedule-depth}
        Given a legal partition with $k$ sets, the \gls{schedule} $\Schd$ the \Gls{Greedy Schedule} creates a scheduling graph $\GSchd$ such that $\Depth{\GSchd} \leq k$.
    \end{lemma}
    \begin{proof}
        \autoref{alg:greedy-schedule} may add directed edges between a vertex of a set $B_j$ to a vertex of a set $B_i$ only if $j < i$.
        Since there are exactly $k$ groups, a path in $\GSchd$ may have at most $k-1$ edges, therefore $\Depth{\GSchd} \leq k$.
    \end{proof}

    \begin{theorem} \label{thm:optimalschedule-nph}
        $\symb{OptimalSchedule}$ is \gls{NPH}.
    \end{theorem}
    \begin{proof}
    We show this using a Cook reduction~\cite{npc-guide} from the $\symb{Color}\in\symb{NPC}$ problem.
    Given an input $\left(G=\left(V, E\right), k\right)$ to the $\symb{Color}$ (\ref{def:color}) problem, we define the following reduction.
    \begin{algorithmic}[1]
        \footnotesize 
        \Function{reduction}{$\left(G, k\right)$}
            \State \textbf{calculate} $\Block=\left(\T,\conflicts, \len\right) \gets$ \Call{transform}{$G$}
            \State \textbf{use} $\symb{OptimalSchedule}$ \textbf{oracle} for $\Block$
            \State $\Schd \gets$ \textbf{oracle answer}
            \If{$\Lt{\Schd} \leq k \cdot c$}
                 \Return \emph{true}
            \Else\
                \Return \emph{false}
            \EndIf
        \EndFunctionExplicit
    \end{algorithmic}

    \paragraph{Polynomial Time.} This follows from \Cref{lemma:transform-poly} and the fact that calculating the depth of a \gls{DAG} can be done in linear time; we conclude that the reduction is computable in polynomial time. $\square$

    \paragraph{Correctness.} By proving that (\ref{eq:reduction-condition}) holds.
    \begin{equation} \label{eq:reduction-condition}
        \left(G,k\right) \in \symb{Color}
        \Longleftrightarrow
        \Call{reduction}{G,k} = \text{true}
    \end{equation}

    \subparagraph{($\Rightarrow$)} Assume that a valid coloring $c:V \to \left\{1, \ldots, k\right\}$ of $G$ exists.
    We then show that for any minimal \gls{schedule} $\Schd$ of the corresponding block $\Block$, the following  $\Lt{\Schd} \leq k \cdot c$ holds, and thus \Call{reduction}{$G,k$} returns \emph{true}.
    \Gls{wlog} assume that the coloring uses exactly $k$ colors.
    Now, denote the set $B_i \subset V=\T$ as the subset of nodes assigned the color $i$, for all colors $1,\ldots, k$.
    Since $c$ is a valid coloring, we know that $\bigsqcup_{i=1}^k B_i = V$ is a partition of $V$ to $k$ disjoint \gls{IS}[s] in $G$.
    Also, we see that according to the \textproc{transaform} function $G=\confGraph{\Block}$, and thus the union $\bigsqcup_{i=1}^k B_i = \T$ is a partition of $\T$ to $k$ \gls{conflict-free} sets.
    Thus, the \gls{schedule} $\Schd'$, created by \gls{Greedy Schedule} for the above partition, is valid according to \Cref{lemma:valid-greedy-schedule}.
    By \Cref{lemma:greedy-schedule-depth}, we conclude that $\Depth{\SchG{\Schd'}} \leq k$; and since $\forall_{tx\in\T}\,\lenOf{tx}=c$ we know that $\Lt[\len]{\Schd'} \leq k \cdot c$.
    Thus, $\Lt{\Schd} \leq k \cdot c$, because the oracle returned the schedule $\Schd$ that has minimal latency. $\square$

    \subparagraph{($\Leftarrow$)} Assume that \Call{reduction}{$G,k$} returns \emph{true}.
    Let $\Schd$ be the valid \gls{schedule} returned by the oracle for which $\Lt{\Schd} \leq k \cdot c$.
    Using, $\Schd$ we can construct the coloring $c:V\to\symb{N+}$ from $\GSchd$.
    We construct the legal coloring using the process described in \Cref{thm:convert-to-partition}.

    The last thing left to show is that $c$ uses at most $k$ different colors.
    Let $t$ be the number of colors used by $c$.
    By construction, it is obvious that some path in $\GSchd$ has exactly $t$ vertices and that there is no longer path with at least $t+1$ vertices in $\GSchd$.
    The sum of the path's vertexes is $t \cdot c$, and by the definition of $\Lt{\Schd}$ we know that $ t \cdot c \leq \Lt{\Schd} \leq k \cdot c$.
    Therefore, $t \leq k$ as~required. $\square$
    \end{proof}

\section{Missing Proofs from \srdsref{sec:lower}}
\label{app:homo-proofs}

\subsection{Proof of \srdsref{thm:lt-numbercolors}}
    \begin{proof}
        Assume, by contradiction, that $\Lt{\Schd_{c_{\min}}} < k$. Hence, by construction, $\Depth{\SchG{\Schd_{c_{\min}}}} < k$.
        We use $\Schd_{c_{\min}}$ as a basis for an alternative coloring, using the same technique as in \Cref{thm:convert-to-partition}.
        Denote by $l$ the largest value assigned to a transaction in this particular process.
        By the assumption that $\Depth{\SchG{\Schd_{c_{\min}}}} < k$ and \Cref{thm:convert-to-partition-num-colors}, we can deduce that $l < k$.
        Thus, we found a contradiction \lightning, since the alternative coloring uses fewer colors than the minimal number of colors needed.
    \end{proof}

\subsection{Proof of \srdsref{thm:lower-bound}}
\begin{proof}
        We know by \Cref{thm:lt-numbercolors} that $\Lt{\Schd_{c_{\min}}} = k$.
        Now assume, by contradiction, that there exists another \gls{schedule} $\Schd'$ whose latency is lower than the \gls{schedule} $\Schd_{c_{\min}}$, and recall that $k$ is the minimal number of colors needed to color $\confGraph{\Block}$.
        Since $\Schd'$ is a (valid) \gls{schedule}, it forms a \gls{DAG} over the set of transactions $\T$ such that there is a directed path between every pair of conflicting transactions.
        We now use $\Schd'$ as a basis for an alternative coloring in a similar fashion to the above lemma.
        Denote $l$ the largest value of the process.
        Thus, it is a valid coloring of $\confGraph{\Block}$ with $l$~colors.

        Also, by definition, since the longest path in $\Schd'$ is shorter than that of $\Schd$ and given the contradiction assumption, $l < k$.
        Therefore, we found a coloring of the conflict graph that uses fewer colors than its minimal coloring, a contradiction. \lightning
    \end{proof}

\section{Proof of \srdsref{thm:minhomolt-eq-chromatic}}
\label{app:minhomolt-eq-chromatic}
The proof of the theorem follows directly from the combination of the following two lemmas:

    \begin{lemma}
        \label{thm:conv-minvc-minlt}
        Consider some block of \gls{homogeneous transactions} represented by the conflict graph $\Block =\left(\T,\conflicts\right)$, and a minimal coloring $c:\T\to\left\{1,\ldots,k\right\}$ of $\confGraph{\Block}$.
        Then, $c$ can be converted to an optimal \gls{schedule} $\Schd$ such that $\Lt{\Schd} = k$, in polynomial time.
    \end{lemma}
    \begin{proofsketch}
        Based on \autoref{sec:lower} we can use the \gls{Greedy Schedule} to create the \gls{schedule} in polynomial time. $\square$
    \end{proofsketch}

    \begin{lemma}
        \label{thm:conv-minlt-minvc}
        Consider some graph $G =\left(V,E\right)$ and denote the induced block of \gls{homogeneous transactions} $\Block$ represented by the conflict graph $\confGraph{\Block} = G$.
        Now consider some optimal homogeneous \gls{schedule} $\Schd$ for $\Block$ such that $\Lt[\mathds{1}]{\Schd} = l$.
        Then, $\Schd$ can be converted to a minimal coloring with $l$ colors in linear time.
    \end{lemma}
    \begin{proof}
        We use the same function \textproc{ConvertToColoring} with $\Schd$ from \Cref{thm:convert-to-partition} and obtain the legal coloring $c:V\to\left\{1, \ldots, k\right\}$.
        This process takes linear time.
        It is left to show that the coloring created from $\Schd$ is minimal, i.e., we need to show that $\chromatic{G} = k$.
        By 
        \Cref{thm:lt-numbercolors} and \Cref{thm:lower-bound} we know that $\chromatic{G}=l$.
        Now, we continue by showing that $l=k$: \Cref{thm:convert-to-partition-num-colors} implies $k = \Depth{\SchG{S}}$.
        Also, $\Depth{\SchG{S}} = \max_{tx \in \SchG{S}} \depth[\unitlen]{tx}$ $= \Lt[\mathds{1}]{\Schd} = l$\ by definition.
        In total $k=l$, and thus $\chromatic{G} = k$.
    \end{proof}

Now, from both \Cref{thm:conv-minvc-minlt} and \Cref{thm:conv-minlt-minvc}, we can now deduce \Cref{thm:minhomolt-eq-chromatic}.

\section{Missing Proofs from \srdsref{sec:homotail-proofs}}
\label{app:homotail-proofs}

\subsection{Proof of \srdsref{thm:reorder-coloring-homo}}
    \begin{proof}
        From \Cref{lemma:greedy-schedule-depth} we can immediately deduce that $\Lt{\Schd} \leq k$.
        The reordered partition $\sqcup_{i=1}^k T_{\sigma\left(i\right)}$ also uses $k$ sets, since it uses the same exact sets as the original one.
        Thus, \Cref{lemma:greedy-schedule-depth} can also be applied for the new partition proving that $\Lt{\Schd_\sigma} \leq k$.
    \end{proof}

\subsection{Proof of \srdsref{thm:reorder-mincoloring-homo}}
    \begin{proof}
        Both $\Schd$ and $\Schd_\sigma$ were created from minimal coloring, so the conditions of \autoref{thm:lower-bound} apply to both of them.
        Thus, we have $\Lt{\Schd} = \MinLt{\T} = \Lt{\Schd_\sigma}$, thereby finishing the proof.
    \end{proof}

    \section{Batch Scheduling}
    \label{sec:batched-schdlr}
    \label{app:batched}
    We present here a simple batch \gls{scheduler}, which divides each block of transactions into subsets of non-conflicting transactions and then executes these subsets one after the other.
    A benefit of the batch \gls{scheduler} is that it does not need any synchronization operations.
    It is especially effective when the transactions are homogeneous and when there are enough cores to execute all transactions belonging to the same subset at the same time.

    \begin{figure}
        \centering
        \includegraphics[width=0.4\linewidth]{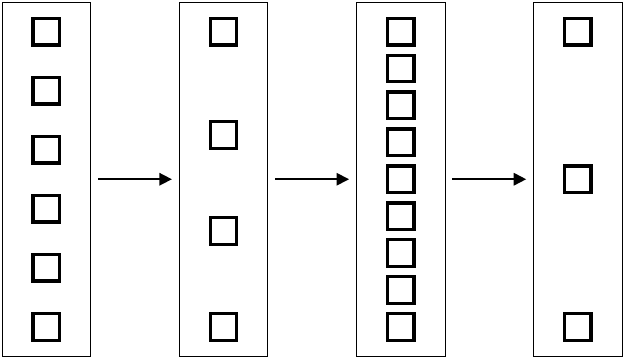}
        \caption{A depiction of a batch-oriented \gls{scheduler}. }
        \label{fig:batched-depiction}
    \end{figure}

    \paragraph{Batch \Gls{schedule}[s].} \autoref{alg:batched:make-schedule} depicts the creation of a batch \gls{schedule}.
    The algorithm requires a partition of the transactions $\symb{disjointunion}\bigsqcup_{i=1}^{k-1} \T_i = \T$ such that each subset is \gls{conflict-free}.
    Similarly to $\GREEDYA$, the partition used must be deterministic in the sense that all partitions used must be equal regardless of the orders of the sets and the transactions. 
    After the partition is obtained, a deterministic order is chosen.
    The partition is returned in that order as the schedule representing a sequence of batches.
    This \gls{scheduler} is denoted~$\BATCHA$.

    \begin{algorithm}
        \caption{Pseudocode for Batch-Oriented \gls{scheduler} $\BATCHA$}
        \label{alg:batched:make-schedule}
        \footnotesize
    \begin{algorithmic}[1]
        \Procedure{MakeSchedule}{Block $\Block=\left(\T, \conflicts\right)$}
            \State Obtain a partition of $\T$ with $k$ groups: $\bigsqcup_{i=1}^{k-1} \T_i = \T$ deterministically
            \Statex \hspace{1em}\ \ \textbf{Require:} $\T_i$ is \glsxtrshort{conflict-free}, i.e., $\forall tx,tx'\in \T_i \: : \: tx \notconflicts tx'$
            \State order $\T_1, \ldots, \T_k$ in some deterministic way \label{alg:batched:make-schedule:ordered-part}
            \State $\Schd \gets \left[\T_1, \ldots, \T_k\right]$
            \State \Return $\Schd$
        \EndProcedure
    \end{algorithmic}
    \end{algorithm}

    \paragraph{Batch Execution.}
    \autoref{alg:batched:execute-schedule} depicts the execution of a batch \gls{schedule}.
    Batches are executed one after the other, while transactions within the same batch are executed concurrently.
    The next batch starts executing only after all transactions in the current batch finish executing.

    \begin{algorithm}[t]
        \caption{Pseudocode for Batch-Oriented \gls{scheduler} $\BATCHA$}
        \label{alg:batched:init-schedule}
        \footnotesize
    \begin{algorithmic}[1]
        \Function{$\GRAPHA\!::$init-execution}{schedule $\Schd = \left[\T_1, \ldots, \T_k\right]$, global-state s}
            \State \textbf{add} dummy batches $\T_0, \T_{k+1}$
            \State temp-store $\gets$ current value of global state s
            \Spawn
                \For{$\T_i = \T_1,\ldots, \T_{k+1}$ ($i$ increases)} \label{alg:batched:execute-schedule:loop}
                    \State \textbf{wait for} all transactions in $\T_{i-1}$ to complete \label{alg:batched:execute-schedule:wait}
                    \State \textbf{schedule} all transactions in $\T_i$ to run \textbf{with} temp-store \textbf{as global state}
                    \State temp-store $\gets$ temp-store $\cup$ (state changes from $\T_i$)
                    \State \textbf{emit} transaction results for those in $\T_i$
                \EndFor
            \EndSpawn
            \State exec $\gets$ ($\T_0$, $\T_{k+1}$, temp-store)
            \State \Return exec
        \EndFunctionExplicit
    \end{algorithmic}
        \label{alg:batched:execute-schedule}
    \begin{algorithmic}[1]
        \Procedure{$\BATCHA\!::$init-execution}{execution e}
            \State \textbf{complete} execution of batch $\T_0$
        \EndProcedure

        \Function{$\BATCHA\!::$is-execution-running}{execution e}
            \If{execution of batch $\T_{k+1}$ completed} \Return true
            \Else \ \Return false
            \EndIf
        \EndFunction

        \Function{$\BATCHA\!::$next-execution-results}{execution e}
            \State \Return transaction results that have been emitted
        \EndFunction

        \Function{$\BATCHA\!::$state-changes}{execution e}
            \State \Return state changes applied to temp-store
        \EndFunction
    \end{algorithmic}
    \end{algorithm}

    \begin{theorem}
        \Gls{scheduler} $\BATCHA$ is sequentially deterministic.
    \end{theorem}
    \begin{proof}
        We prove this theorem by showing the criteria in \autoref{thm:detr-critirion2}.
        Fix some block~$\Block$.
        
        {(\texttt{serializable})} We first prove that each execution of $\BATCHA$ for the given block $\Block$ is \gls{serializable}.
        Consider some execution $\sigma$ of $\BATCHA$ for $\Block$.
        Let the schedule $\Schd$ be the batch sequence $[\T_1,...,\T_k]$ that \autoref{alg:batched:make-schedule} returned for the block $\Block$.
        Let $\tau$ be a sequential execution obtained by extending the partial order of transactions in $\sigma$ to a total order such that all transactions of the same batch are placed one after the~other.

        By definition, all transactions in the same batch $\T_i$ are non-conflicting.
        Hence, the values read and written by such transactions are independent of their relative order and, in particular, are the same in both $\tau$ and in $\sigma$.
        Additionally, by construction, all conflicting transactions in $\Block$ are ordered in the same manner in both $\tau$ and $\sigma$.
        In summary, $\tau$ is a valid sequential execution that is \emph{conflict equivalent} to $\sigma$, and therefore $\sigma$ is \gls{serializable}~\cite{DB-textbook} as required.

            {(\texttt{deterministic})} Next, we show that $\BATCHA$ is deterministic.
        Consider two executions $\sigma$ and $\sigma'$ of $\BATCHA$ for $\Block$.
        Since \autoref{alg:batched:make-schedule} is deterministic, the above partition is the same partition used for $\sigma$ and $\sigma'$.
        Let two conflicting transactions $tx_i \in \T_i$  and $tx_j \in \T_j$.
        Since $tx_i \conflicts tx_j$, we deduce that $\T_i \neq \T_j$.
        \Gls{wlog}, assume that $i < j$.
        The loop~(line~\ref{alg:batched:execute-schedule:loop}) on line~\ref{alg:batched:execute-schedule:wait} causes the \gls{scheduler} to wait for batch $\T_i$ to finish completely before starting to execute batch $\T_{i+1}$ (or higher).
        In particular, all transactions in batch $\T_j$ do not start before all transactions in batch $\T_i$ finish, in both $\sigma$ and $\sigma'$.
        This includes both $tx_i$ and $tx_j$; therefore, $tx_i$ ends before $tx_j$ starts in both $\sigma$ and $\sigma'$.
        Hence, all conflicting transactions in the same block are ordered in the same manner in $\sigma$ and $\sigma'$.
        For similar arguments as above, for all non-conflicting transactions, the values read and written by them are independent of their relative order, and in particular are the same in $\sigma$ and in $\sigma'$.
        In summary, $\sigma$ is \emph{conflict equivalent} to $\sigma'$, thus $\sigma \equiv \sigma'$ as necessary.
    \end{proof}

    \subsection{Latency and Optimality}
    In \autoref{sec:graph-latency} we introduced the latency property for a graph schedule and demonstrated how it captures the abstract concept of block execution duration.
    Then, in \autoref{chap:optimal} we defined what an optimal schedule is, based on its latency.
    We can introduce corresponding definitions for batch schedules in a similar fashion.
    
    \begin{definition}[Latency (Batch Schedules)]
        The latency of a batch \gls{schedule} $\Schd = \left[B_1, \ldots, B_k\right]$ is the sum of the length of its batches, denoted:
        \begin{equation}
            \BLt[\len]{\Schd} \triangleq
            \sum_{i=1}^k \lenOf{B_i} = 
            \sum_{i=1}^k \max_{tx \in B_i} \lenOf{tx}
        \end{equation}
        when the length of a batch is the length of its longest transaction, $\lenOf{B} \triangleq \max_{tx \in B} \lenOf{v}$.
    \end{definition}
    
    One can extend the principles we used to show how the latency properly represents the execution duration of graph schedules using $\GRAPHA$, to batch schedules and $\BATCHA$.
    The primary reason for this is that in $\BATCHA$ batches execute non-concurrently and wait until the previous batch finishes. Of course, the execution of a batch finishes at least after its longest transaction ends.

    Similarly, the definitions for optimal batch \gls{schedule}[s] and for optimal latency are identical to the versions for graph \gls{schedule}[s].
    We use the notation $\MinBLt[\len]{\Block}$ to denote the optimal (batch) latency of the block $\Block$.

    \subsection{Relation to Graphed Execution} \label{sec:batch-graph-formality}
    \begin{figure}
        \centering
        \subfloat[\normalfont{Original Batch Scheduler}]{
            \includegraphics[width=0.45\linewidth]{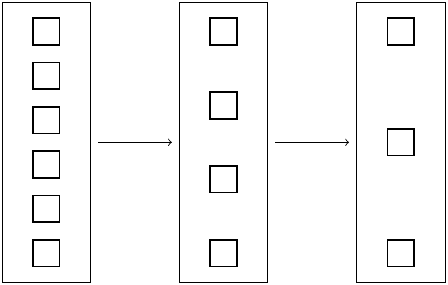}
        }
        \hfill
        \subfloat[\normalfont{Equivalent Batch Scheduler}]{
            \includegraphics[width=0.45\linewidth]{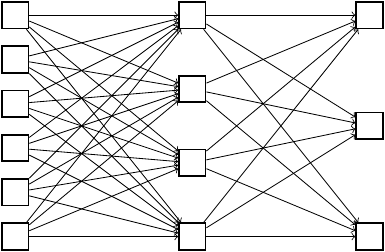}
        }
        \caption{A depiction of the conversion between a batch-oriented scheduler to an equivalent scheduling-graph-oriented scheduler.}
    \end{figure}
    Obviously, batch schedules, i.e., schedules produced by \autoref{alg:batched:make-schedule}, do not strictly follow the formal definition of a graph \gls{schedule} (\autoref{def:schedule}).
    However, we can show that a batched \gls{schedule} $[\T_i]= [\T_1,\ldots,\T_k]$ can be transformed into a corresponding graph schedule $\Schd_{[\T_i]}$ as shown in \autoref{eq:batch-edges}.
    \begin{equation} \label{eq:batch-edges}
    \small 
    \Schd_{[\T_i]} = \bigsqcup_{i=1}^{k-1}\left\{\left(tx_i, tx_{i+1}\right) \mid tx_i \in \T_i \land tx_{i+1} \in \T_{i+1} \right\} =  \bigsqcup_{i=1}^{k-1} \T_i \times \T_{i+1}
    \end{equation}
    When $\Schd_{[\T_i]}$ is combined with $\GRAPHA$, we achieve an equivalent behavior to that of $\BATCHA$.
    Execution results are the same for both schedules, since the construction of $\Schd_{[\T_i]}$ forces that only transactions from the same batch may run concurrently and that the batch order is respected.
    The latency is also preserved in this transformation, that is, $\BLt[\len]{[\T_i]} = \Lt[\len]{\Schd_{[\T_i]}}$, since every path of $k$ transactions in $\Schd_{[\T_i]}$ can be correlated with a selection of the \emph{same} $k$ transactions from batches, so each is selected from a different batch.
    Maximizing the length of the path is actually maximizing the length of the selected transactions, that is, selecting the longest transaction from each batch.
    This is also true vice versa.

    The explanation above allows us to view batch scheduling as a restricted version of graph scheduling, where only batch schedules are considered.
    This also helps us to immediately deduce properties of batch scheduling, for example $\MinLt{\Block} \leq \MinBLt{\Block}$ for any block.
    
    \subsection{Homogeneous Blocks}
    We now discuss the feasibility of optimal batch scheduling for homogenous blocks using a minimal coloring.
    Here, we build on the correctness of the $\MINCOLORA$ \gls{scheduler} from \autoref{chap:homo}.

    Assume that we have some homogeneous block $\Block$ and some optimal graph schedule $\Schd$ of it.
    We transform this schedule into a batch sequence by grouping transactions by their depths, i.e., $tx$ with $\depth[\Schd]{tx}=i$ goes to batch $\T_i$.
    Since all transactions have the same execution lengths, this new batch schedule has the same latency as the original schedule $\Schd$.
    Thus, it is also optimal as a batch schedule.

    The above transformation process is similar to the procedure $\textproc{ConvertToColoring}$ (\autoref{alg:converttocoloring}).
    We use this relationship to explain why a minimal coloring also helps create an optimal batch schedules as well when the block is homogenous.

\end{document}